\documentclass[12pt]{article}
\usepackage{color}
\usepackage{amsfonts}
\usepackage{amsthm}
\usepackage{amssymb,bm}
\usepackage{amsmath,enumerate,rotate}
\usepackage{amssymb,latexsym,mathrsfs}
\usepackage{graphicx}

\RequirePackage[OT1]{fontenc}
\RequirePackage{amsthm,amsmath}
\RequirePackage[numbers]{natbib}
\RequirePackage[colorlinks,citecolor=blue,urlcolor=blue]{hyperref}
\usepackage[font=footnotesize]{caption}
\newcommand{\ddr}{d}
\newcommand{\edr}{\mathrm{e}}
\newcommand{\E}{\mathbb{E}}

\renewcommand{\P}{\mathbb{P}}
\newcommand{\X}{\mathbb{X}}
\def\I{\mbox{I}}

\newcommand{\R}{\mathbb{R}}
\newcommand{\N}{\mathds{N}}

\newcommand{\Ga}{\Gamma}

\newcommand{\Bcr}{\mathcal{B}}

\newcommand{\Xcr}{\mathcal{X}}

\def\N{\mbox{N}}

\numberwithin{equation}{section}
\theoremstyle{plain}
\newtheorem{thm}{Theorem}[section]

\newtheorem{coro}{Corollary}[section]

\theoremstyle{definition}
\newtheorem{dfn}{Definition}
\theoremstyle{remark}
\newtheorem*{rmk}{Remark}
\begin{document}

\title{Compound random measures and their use in Bayesian nonparametrics}

\author{Jim E. Griffin
and Fabrizio Leisen\footnote{{\it Corresponding author}: Jim E. Griffin, School of Mathematics, Statistics and Actuarial Science, University of Kent, Canterbury CT2 7NF, U.K. {\it Email}: 
jeg28@kent.ac.uk}\\
University of Kent}

\date{}



\maketitle

\abstract{
A new class of dependent random measures which we call {\it compound random measures} are proposed and the use of normalized versions of these random measures as priors in Bayesian nonparametric mixture models is considered. Their tractability allows the properties of both compound random measures and normalized compound random measures to be derived. In particular, we  show how compound random measures can be constructed with gamma, $\sigma$-stable and generalized gamma process marginals. We also derive several forms of the Laplace exponent and characterize dependence through both the L\'evy copula and correlation function. A slice sampler and an augmented P\'olya urn scheme sampler are described for posterior inference  when a normalized compound random measure is used as the mixing measure in a nonparametric mixture model and a data example is discussed.}
\\
\noindent\textbf{Keyword}: Dependent random measures; L\'evy Copula; Slice sampler; Mixture models; Multivariate L\'evy measures; Partial exchangeability.



\section{Introduction}

Bayesian nonparametric mixtures have become a standard tool for inference when a distribution of either observable or unobservable quantities is considered unknown. A more challenging problem, which arises in many applications, is to define a prior for a collection of related unknown distributions. For example, \cite{mulros04} consider informing the analysis of a study with results from previous related studies. They considered the CALGB 9160 \citep{CALGB9160} clinical study which  looked at the response over time of patients to different anticancer drug therapies.  \cite{mulros04} suggested improving the precision of their inference using the results of the related  study
CALGB 8881\citep{CALGB8881}.
Figure~\ref{f:data1} shows bivariate plots of two subject-specific regression parameters ($\beta_0$ and $\beta_1$) for the two studies.
\begin{figure}[h!]
\begin{center}
\includegraphics[trim=0mm 0mm 120mm 230mm, clip]{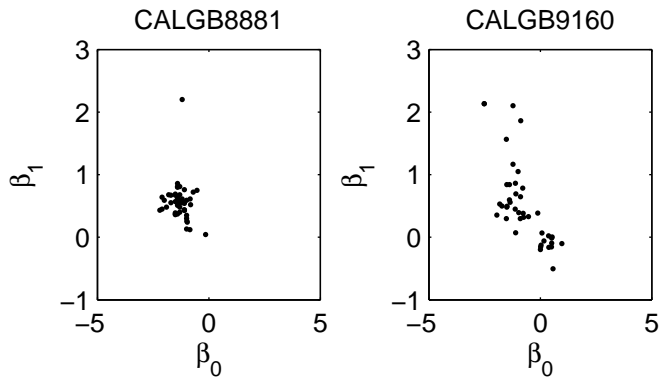}
\end{center}
\caption{Scatter-plots of the subject-specific regression parameters  $\beta_0$ and $\beta_1$ for the groups in CALGB 8881 and CALGB 9160.}\label{f:data1}
\end{figure}
 The graphs suggest differences between the joint distribution of $\beta_0$ and $\beta_1$ which should be included in any analysis which combines these data sets. The results for CALGB9160 also suggest that a nonparametric model is needed to fully describe the shape of the density. A natural Bayesian approach would assume different distributions for each study but construct a dependent prior for  these distributions. 

In general, suppose that $x\in\mathcal{X}$ denotes the value of covariates then, in a Bayesian nonparametric analysis, a prior needs to be defined across a collection of correlated distributions $\{\tilde{p}_x\vert x\in\mathcal{X}\}$. This problem was initially studied in a seminal paper on dependent Dirichlet processes \citep{MacEachern} where generalisations of the Dirichlet process were proposed. Subsequent work used stick-breaking constructions of random measures as a basis for defining such a prior. This work is reviewed by \cite{dun10}. These priors can usually be represented as
\begin{equation}
\tilde{p}_x=\sum_{i=1}^{\infty} w_i(x)\delta_{\theta_i(x)}
\label{DDP}
\end{equation}
where 
$w_1(x),w_2(x),\dots$ follow a stick-breaking process for all $x\in\mathcal{X}$. 
A drawback with this approach is the stochastic ordering of the $w_i(x)$'s for any $x\in\mathcal{A}$ which can lead to strange effects in the prior as $x$ varies.

If $\mathcal{A}$ is countable,  several other approaches to defining a prior on a collection of random probability measures have been proposed. The 
Hierarchical Dirichlet process (HDP) \citep{tehjor06} assumes that $\tilde{p}_x$ are {\it a priori} conditionally  independent and identically distributed according to a Dirichlet process whose centring measure is itself given a Dirichlet process prior. This construction induces correlation between the elements of $\{\tilde{p}_x\vert x\in\mathcal{A}\}$ in the same way as in parametric hierarchical models. This construction can be extended to  more general hierarchical frameworks 
\cite[see {\it e.g.}][for a review]{tehjord}. Alternatively, a prior can be defined using the idea of normalized random measures with independent increments which are  defined by normalising a completely random measure. The prior is defined on a collection of correlated completely random measures $\{\tilde\mu_x\vert x\in\mathcal{A}\}$  which are then normalized for each  of $x$, {\it i.e.} $\tilde{p}_x=\tilde\mu_x/\tilde\mu_x(\mathbb{X})$ where $\mathbb{X}$ is the support of $\tilde\mu_x$. Several specific constructions have been proposed including various forms of superposition \citep{GKS13, NL14,NLP14a,NLP14b, Chen13, BaCaLe},  the kernel-weighted completely random measures \citep{fotwil12, Gri11} and L\'evy copula-based approaches \citep{LL, LLS, ZL}. In this paper, we develop an alternative method for constructing correlated completely random measures which is tractable, whose properties can be derived and for which sampling methods  for posterior inference without truncation can be developed. The construction also provides a unifying framework for previously proposed constructions. Indeed, the $\sigma$-stable and gamma vector of dependent random measures, studied in the recent works of \cite{LL},  \cite{LLS} and \cite{ZL} are special cases. Although these papers derive useful  theoretical results, their application has been limited by the lack of a sampling methods for   posterior inference. The algorithms proposed in this paper can also be used for posterior sampling for these nonparametric priors which is another contribution of the paper. 

The paper is organized as follows. Section 2 introduces the concepts of completely random measures, normalized random measures and their multivariate extensions. Section 3 discusses the construction and some properties of a new class of multivariate L\'evy process, \emph{Compound Random Measures}, defined by a {\it score distribution} and a {\it directing L\'evy process}. Section 4 provides a detailed description of Compound Random Measures with a gamma score distribution. Section 5 considers the use of normalized version of Compound Random Measures in nonparametric mixture models including the description of a Markov chain Monte Carlo scheme for inference. Section 6  provides an illustration of the use of these methods in an example and Section 7 concludes.

\section{Preliminaries}
Let $(\Omega,\mathcal{F},\P)$ be a
probability space and $(\X,\mathcal{X})$ a measure space, with $\X$ 
Polish and $\mathcal{X}$ the Borel $\sigma$--algebra of subsets of $\X$. Denote by $\mathbb{M}_{\mathbb{X}}$ the space of boundedly finite measures on $(\X,\mathcal{X})$, {\it i.e.} this means that for any $\mu$ in $\mathbb{M}_{\mathbb{X}}$ and any bounded set $A$ in $\mathcal{X}$ one has $\mu(A)<\infty$. Moreover, $\cal{M}_{\mathbb{X}}$ stands for the corresponding Borel $\sigma$--algebra, see \cite{daley} for technical details. 
The concept of a \textit{completely random measure} was introduced by
\cite{Kingman67}.
\begin{dfn}
Let $\tilde{\mu}$ be a measurable mapping from $(\Omega,\mathcal{F},\P)$ into ($\mathbb{M}_{\mathbb{X}}$,$\cal{M}_{\mathbb{X}}$) and such that for any $A_1,\dots,A_n$ in $\mathcal{X}$, with $A_i\cap A_j=\emptyset$ for any $i\neq j$, the random variables $\tilde{\mu}(A_1), \dots, \tilde{\mu}(A_n)$ are mutually independent. Then $\tilde{\mu}$ is called a \textit{completely random measure} (CRM). 
\end{dfn}
\noindent A CRM can always be represented as a sum of two components: 
$$\tilde{\mu}=\tilde{\mu}_c+\sum_{i=1}^{M}V_i\delta_{x_i}$$
where the fixed jump points $x_1,\dots, x_M$ are in $\mathbb{X}$ and the non-negative random jumps $V_1,\dots,V_M$ are both mutually independent and independent from $\tilde{\mu}_c$. The latter is a completely random measure such that
$$\tilde{\mu}_c=\sum_{i=1}^{\infty} J_i\delta_{X_i}$$
where both the positive jump heights $J_i$'s and the $\mathbb{X}$-valued jump locations $X_i$'s are random. The measure $\tilde{\mu}_c$ is characterized by the \textit{L\'evy-Khintchine} representation which states that
$$\E\left[e^{-\int_X f(x)\tilde{\mu}_c(dx)}\right]
=e^{-\int_0^{\infty}\int_X [1-e^{-sf(x)}]\bar\nu(\ddr s,\ddr x)}
$$
where $f:\mathbb{X}\rightarrow\mathbb{R}^{+}$ is a measurable function such that $\int f\tilde{\mu}_c<\infty$ almost surely and $\bar\nu$ is a measure on $\mathbb{R}^{+}\times\mathbb{X}$ such that
$$
\int_{\mathbb{R}^{+}}\int_B \min\{1,s\}\bar\nu(\ddr s,\ddr x)<\infty$$
for any $B$ in $\mathcal{X}$. The measure $\bar\nu$ is usually called the L\'evy intensity of $\tilde{\mu}_c$. Throughout the paper, we will consider completely random measures without the fixed jump component ({\it i.e.} $M=0$).
For our purposes, we will focus on the \textit{homogeneous} case, {\it i.e.} L\'evy intensities where the height and location contributions are separated. Formally,
$$ \bar\nu(\ddr s,\ddr x)=\rho(\ddr s)\alpha(\ddr x)$$
where  $\rho$ is a measure on $\mathbb{R}^+$ and
$\alpha$ is a non-atomic measure on $\mathbb{X}$, which is usually called the {\it centring measure}. Some famous examples are  the \textit{Gamma} process,
$$\bar\nu(\ddr s,\ddr x)=s^{-1}e^{-s}\ddr s\,\alpha(\ddr x),$$
 the \textit{$\sigma$-stable} process,
$$\bar\nu(\ddr s,\ddr x)=\frac{\sigma}{\Gamma(1-\sigma)}s^{-1-\sigma} \ddr s\,\alpha(\ddr x),\quad 0<\sigma<1,$$
and the homogeneous \textit{Beta} process,
\[
\bar\nu(\ddr s,\ddr x)= \theta s^{-1}(1-s)^{\theta-1}\ddr s\,\alpha(\ddr x),\quad 0<s<1,\quad \theta>0.
\]
A general class of processes that includes the gamma and $\sigma$-stable process is the \textit{Generalized Gamma} process,
$$\bar\nu(\ddr s,\ddr x)=\frac{\sigma}{\Gamma(1-\sigma)}s^{-1-\sigma} e^{-as}\ddr s\,\alpha(\ddr x),\quad 0<\sigma<1,\quad a>0$$
Random measures are the basis for building Bayesian nonparametric priors. 
\begin{dfn}
Let $\tilde{\mu}$ be a measure in ($\mathbb{M}_{\mathbb{X}}$,$\cal{M}_{\mathbb{X}}$). A \textit{Normalized Random Measure} (NRM) is defined as $\tilde{p}=\frac{\tilde{\mu}}{\tilde{\mu}(\mathbb{X})}$.
\end{dfn}
\noindent The definition of a normalized random measure is very general and does not require that the underlying measure is completely random. The Pitman-Yor process (see \cite{PY97}) is
a well-known example of a Bayesian nonparametric priors which cannot be derived by normalizing a completely random measure. In this particular case, the unnormalized measure is obtained through a change of measure of a $\sigma$-stable process.  
However, many common Bayesian nonparametric priors can be defined as a normalization of a CRM and many other processes can be derived by normalising processes derived from CRMs, see \cite{rlp}. For instance, 
it can be shown that the \textit{Dirichlet Process}, introduced by \cite{ferg}, is a normalized gamma process. Throughout the paper, we will assume that the underlying measure is a CRM and use the acronym NMRI (Normalized Random Measures with independent increments) to emphasize the independence of a CRM on disjoint intervals.

Although nonparametric priors based on normalization are extremely flexible, in many
real applications data arise under different conditions and hence
assuming a single prior can be too restrictive. For example,
using covariates, data may be divided into different units. In this
case, one would like to consider different distributions for different
units instead of a single common distribution for all the units. In these situations, it is more reasonable to consider vectors of dependent random probability measures.  
   
\subsection{Vectors of normalized random measures} Suppose $\tilde\mu_1,\dots, \tilde\mu_d$ are homogeneous CRMs on $(\X,\mathcal{X})$
 with respective marginal L\'evy intensities
\begin{equation}\label{Homo}
\bar{\nu}_j(\ddr s, \ddr x)=\nu_j(\ddr s)
\,\alpha(\ddr x),
\qquad \quad j=1,\dots,d.
\end{equation}
where  $\nu_j$ is a measure on $\mathbb{R}^+$ and
$\alpha$ is a non-atomic measure on $\mathbb{X}$. Furthermore, $\tilde\mu_1, \dots, \tilde\mu_d$ are dependent and the random vector
$(\tilde\mu_1,\dots,\tilde\mu_d)$ has independent increments, in the sense that for any $A_1,\dots,A_n$ in $\mathcal{X}$, with $A_i\cap A_j=\emptyset$ for any $i\neq j$, the random vectors $(\tilde\mu_1(A_i),\dots,\tilde\mu_d(A_i))$ and
$(\tilde\mu_1(A_j),\dots, \tilde\mu_d(A_j))$ are independent. This implies that
for any set of measurable functions $\bm{f}=(f_1,\dots,f_d)$ such that $f_j:\X\to\mathbb{R}^{+}$, $j=1,\dots,d$ and $\int |f_j|\,d\tilde{\mu}_j<\infty$, one has a multivariate analogue of the L\'evy-Khintchine representation (see \cite{sato}, \cite{daley} and \cite{epifani})
    \begin{equation}
    \label{eq:biLaplace}
    \E\left[e^{-\tilde\mu_1(f_1)-\cdots-\tilde\mu_d(f_d)}\right]=
    \exp\left\{-\psi^*_{\rho,d}(\bm{f})\right\}
    \end{equation}
where $\tilde\mu_j(f_j)=\int f_j\,d\tilde{\mu}_j$,
\begin{equation}
\label{ExpLevy}
\psi^*_{\rho,d}(\bm{f})=\int_\X\int_{(0,\infty)^d} 
      \left[1-\edr^{-s_{_1} f_1(x)-\cdots-s_{d} f_d(x)}\right]\;\rho_d(\ddr s_1,\dots,\ddr s_d)
    \:\alpha(\ddr x)
\end{equation}    
    and
\begin{equation}
    \label{eq:marginals1}
    \int_{(0,\infty)^{d-1}}\rho_d(\ddr s_1,\dots,ds_{j-1},A,ds_{j+1},\dots,ds_d)=
        \int_A \nu_j(ds).
    \end{equation}

The representation (\ref{Homo}) implies that the jump heights of $(\tilde\mu_1,\dots,\tilde\mu_d)$ are independent from the jump locations. Moreover, these jump locations are common to all the CRMs and are governed by $\alpha$. It is worth noting that, since $(\tilde\mu_1,\ldots,\tilde\mu_d)$ has independent increments, its distribution is characterized by a choice of $f_1,\ldots,f_d$ in \eqref{eq:biLaplace} such that $f_j=\lambda_j\,1_{A}$ for any set $A$ in $\Xcr$, $\lambda_j\in\R^+$ and $j=1,\ldots,d$. In this case
\[
\psi^*_{\rho,d}(\bm{f})=\alpha(A)\,\psi_{\rho,d}(\bm{\lambda})
\]
where $\bm{\lambda}=(\lambda_1,\ldots,\lambda_d)$ and
\begin{equation}
\psi_{\rho,d}(\bm{\lambda})=\int_{(\R^+)^d}
\left[1-\edr^{-\langle\bm{\lambda},\bm{s}\rangle}
\right]\rho_d(ds_1,\dots,ds_d)
\label{eq:lapl_exp_star}
\end{equation} 
where $\bm{s}=(s_1,\ldots,s_d)$ and $\langle\bm{\lambda},\bm{s}\rangle=\sum_{j=1}^d \lambda_j s_j $.

We close the section with the definition of \textit{vectors of normalized random measures with independent increments}. 
\begin{dfn}\label{DefVec}
Let ($\tilde\mu_1,\dots, \tilde\mu_d$) be a vector of CRMs on $\mathbb{X}$ and let $\tilde{p}_j=\frac{\mu_j}{\mu_j(\mathbb{X})}$, $j=1,\dots,d$. The vector 
\begin{equation}
\label{vettore}
\tilde{p}=(\tilde
p_1,\dots,\tilde p_d)
\end{equation} 
is called a \textit{vector of dependent normalized random measures with independent increments} on
$(\X,\Xcr)$. 
 \end{dfn}

\section{Compound Random Measures}
In this section, we will define a general class of vectors of NRMI that incorporates  
 many recently proposed priors built using normalization, see  for instance \cite{LL}, \cite{LLS}, \cite{ZL}, \cite{GKS13} and \cite{NLP14a}. Before introducing the formal definition of Compound Random Measures, we want to provide an intuitive illustration of the model. Consider the following dependent random probability measures:
$$\tilde{p}_1=\sum_{i\geq 1} \pi_{1,i}\delta_{X_i}, \,\dots, \,\tilde{p}_d=\sum_{i\geq 1} \pi_{d,i}\delta_{X_i},$$ 
where
\begin{equation}\label{Corm:reviewer}
\pi_{j,i}=\frac{m_{j,i} J_i}{\sum_l m_{j,l} J_l}. 
\end{equation}
The $m_{j,i}$'s are perturbation coefficients that identify specific features of the $j$-th random measure and they are independent and identically distributed across the random measures. The shared jumps $(J_i)_{i\geq 1}$ lead to  dependence among the $\tilde{p}_j$. In the next section, we will provide a formal definition of Compound random measures in terms of its multivariate L\'evy intensity. 

\subsection{Definition} Let $(\tilde\mu_1,\dots, \tilde\mu_d)$ be a vector of homogeneous CRMs on $\mathbb{X}$, {\it i.e.} the L\'evy intensity $\nu_j$ of the measure $\tilde{\mu}_j$ is
$$ \bar\nu_j(ds,dx)=\nu_j(ds)\,\alpha(dx),\quad j=1,\dots,d.$$
Following the notation in Eq. \eqref{ExpLevy}, we want to define a $\rho_d$ such 
\begin{equation}
    \label{eq:marginals}
    \int_{(0,\infty)^{d-1}}\rho_d(\ddr s_1,\dots,ds_{j-1},A,ds_{j+1},\dots,ds_d)=
        \int_A \nu_j(ds)
    \end{equation}
     for any $j=1,\dots,d$. In this setting we can define a compound random measure. 
\begin{dfn}
A  \textit{Compound random measure} (CoRM) is a vector of CRMs defined by a 
 \emph{score distribution} $h$ and a
 \emph{directing L\'evy process} with intensity $\nu^*$ such that
\begin{equation}\label{compound}
\rho_d(ds_1,\dots,ds_d)=\int  h(s_1,\dots,s_d\vert z)\,ds_1\cdots ds_d\,\nu^{\star}(dz)
\end{equation}
where $h(\cdot\vert z)$ is the probability mass function or probability density function of the score distribution with parameters $z$ and $\nu^{\star}$ is the L\'evy intensity of the directing L\'evy process  which satisfies the condition
$$\int\int\min(1,\parallel\bm{s}\parallel) 
h(s_1,\dots,s_d\vert z)\,d\bm{s}
\,\nu^{\star}(dz)
<\infty$$
where $\parallel\bm{s}\parallel$ is the Euclidean norm of the vector $\bm{s}=(s_1,\dots,s_d)$.
\end{dfn}

\medskip

The compound Poisson process with jump density $h$ is a compound random measure with a score density $h$ and  whose directing L\'evy process  is  a Poisson process. Therefore, compound random measures can be seen as a generalisation of compound Poisson processes.
It is straightforward to show that $\tilde{\mu}_1,\dots,\tilde{\mu}_d$  can be expressed as
\begin{equation}
\tilde{\mu}_j=\sum_{i=1}^{\infty} m_{j,i}J_i\delta_{X_i}
\label{jump_mug}
\end{equation}
where $m_{1,i},\dots,m_{d,i}\stackrel{i.i.d.}{\sim} h$ are {\it scores} and
\[
\tilde{\eta}=\sum_{i=1}^{\infty} J_i\delta_{X_i}
\]
is a CRM with L\'evy intensity $\nu^{\star}(ds)\alpha(dx)$. This makes the structure of the prior much more explicit. The random measures share the same jump locations (which have distribution $\alpha/\alpha(\mathbb{X})$)
 but the $i$-th jump has a height  $m_{j,i}J_i$  in
the 
 $j$-th measure and so the jump heights are re-scaled by the score (a larger score implies a larger jump height).
Clearly, the shared factor $J_i$ leads to dependence between the jump heights in each measure. 

\medskip

The construction can be seen in an alternative way, in terms of (augmented) dependent Poisson random measures. Indeed, 

$$\mathbb{E}\left[e^{-\lambda \tilde{\mu}_j(A)}\right]=e^{-\alpha(A)\int_0^{\infty}(1-e^{-\lambda s})\int h_j(s|z)\nu^{*}(dz)ds}.$$
If in the exponent we set $s=mz$, then 
$$\mathbb{E}\left[e^{-\lambda \tilde{\mu}_j(A)}\right]=e^{-\alpha(A)\int_0^{\infty}(1-e^{-\lambda mz})\int z h_j(mz|z)\nu^{*}(dz)dm}$$
which entails that $\tilde{\mu}_j(dx)=\int_0^{\infty}\int mz N_j(dm,dz,dx)$ where 
$$N_j=\sum_{i\geq 1} \delta_{(m_{j,i},J_i,X_i)}$$
is a Poisson random measure with intensity on $(0,+\infty)^2\times\mathbb{X}$ and with L\'evy intensity given by $\alpha(dx)z h_j(mz|z)\nu^{*}(dz)$. This is identical to the distribution given by \eqref{jump_mug}. The Poisson processes $(N_1,\dots,N_d)$ are dependent because they share $\{(J_i,X_i): i\geq 1\}$. The term ``augmentation'' here refers to the fact that the Poisson random measures that characterize the CRMs are typically on $(0,+\infty)\times \mathbb{X}$. A third dimension is introduced to account the heterogeneity across different measures.

\bigskip

To ensure the existence of the vectors of normalized CoRM, as introduced in Definition \ref{DefVec}, the following condition must be satisfied for each $j=1,\dots,d$: 
$$\nu_j((0,+\infty))=\int_0^{+\infty}\int h_j(s|z)\nu^{\star}(dz)ds=+\infty$$
where $h_j(s|z)=\int h(s_1,\dots,s_{j-1},s,ds_{j+1},\dots,s_{d}|z)ds_1\cdots ds_{j-1}ds_{j+1}\cdots ds_{d}$. If this condition does not hold true, then $\tilde{\mu}_j(\mathbb{X})=0$ with positive probability and the normalization does not make sense, see \cite{rlp}. 

In this paper, we will concentrate on the sub-class of CoRMs with a continuous score distribution which has independent dimensions and a single scale parameter so that
\[
h(s_1,\dots,s_d\vert z)=z^{-d}\prod_{j=1}^d  f(s_j/z)\,
\]
where $f$ is a univariate distribution. This implies that each marginal process has the same 
 L\'evy intensity of the form
 \begin{equation}
\nu_j(ds)=\nu(ds)=\int z^{-1}f(s\vert z)\,ds\,\nu^{\star}(dz).
\label{link_eqn}
\end{equation}
In Section~\ref{sec:comp_method},
 algorithms are introduced to sample from the posterior of a hierarchical mixture models whose parameters are driven by a vector of normalized compound random measures. These samplers depend crucially on knowing the form of the Laplace Exponent and its derivatives. 
 Some general results about the Laplace exponent and the dependence are available if we assume that the density $z^{-1}f(s_i/z)$ admits a moment generating function.
\begin{thm}
Let
$$M_z^f(t)=\int e^{ts}z^{-1}f(s/z)ds$$
be the moment generating function of $z^{-1}f(s_j/z)$ and suppose that it exists. 
Then
\begin{equation}\label{MGM}
\psi_{\rho,d}(\lambda_1,\dots,\lambda_d)=\int \left(1-\prod_{j=1}^{d} M_z^f(-\lambda_j)\right)\nu^{\star}(z)dz.
\end{equation}
\label{mgf_thm}\end{thm}

The proof of the Theorem stated above is in the appendix as well as a further result about the derivatives of the Laplace exponent. 

\section{CoRMs with independent gamma distributed scores}

In this paper, we will focus on exponential or gamma score distributions.
Throughout the paper we will write $\mbox{Ga}(\phi)$ to be  a gamma distribution (or density) with shape $\phi$ and mean $\phi$ which has density
\begin{equation}\label{Jumps}
f(x)=\frac{1}{\Gamma(\phi)}x^{\phi-1}\exp\{-x\}.
\end{equation}
 This implies that
 $z^{-1}f(y/z)$ is the density of a gamma distribution with shape parameter equal to $\phi$ and mean $\phi\, z$. 
The L\'evy intensities $\nu$ and $\nu^{\star}$ and the score density $f$ are linked by (\ref{link_eqn}) and a CoRM can be defined by either deriving $\nu^{\star}$ for a fixed choice of  $f$ and $\nu$ or by directly specifying $f$ and $\nu^{\star}$. In this latter case, it is interesting to consider the properties of the induced $\nu$.

Standard inversion methods can be used to derive the form of $\nu^*$. Equation \eqref{link_eqn} implies that
$$\nu(s)=\int z^{-1}\frac{1}{\Gamma(\phi)}\left(\frac{s}{z}\right)^{\phi-1}\exp\left(-\frac{s}{z}\right)\nu^*(z) dz$$
The change of variable $t=z^{-1}$ leads to 
$$\nu(s)=\frac{s^{\phi-1}}{\Gamma(\phi)}\int \exp\left(-st\right)t^{\phi-2} \nu^*\left(\frac{1}{t}\right) dt.$$
The above integral can be seen as the classical Laplace transform of the function $f(t)=t^{\phi-2} \nu^*\left(\frac{1}{t}\right)$. If we denote by $\mathcal{L}$ the Laplace transform then
$$\nu(s)=\frac{s^{\phi-1}}{\Gamma(\phi)}\mathcal{L}(f(t))(s)$$
This means that
$$\nu^*\left(\frac{1}{t}\right)=t^{2-\phi}\mathcal{L}^{-1}\left(\frac{\Gamma(\phi)}{s^{\phi-1}}\nu(s)\right)(t)$$
where $\mathcal{L}^{-1}$ is the inverse Laplace transform. This ensures the unicity of $\nu^{*}$.
The forms for some particular choices of marginal process are shown in Table~\ref{t:derived1}.
\begin{table}[h!]
\begin{center}
\begin{tabular}{rll}\hline
$\nu^*(z)$ & Support & Marginal Process\\ \hline
$z^{-1}(1-z)^{\phi-1}$ & $0<z<1$&$\mbox{Gamma}$\\
& \\
$z^{-\sigma-1}\frac{\Gamma(\phi)}{\Gamma(\phi+\sigma)\Gamma(1-\sigma)}$ & $z>0$&$\mbox{$\sigma$-stable}$\\
&\\
$
\frac{\sigma\Gamma(\phi)}{\Gamma(\phi+\sigma)\Gamma(1-\sigma)} z^{-\sigma-1}(1-a\,z)^{\sigma+\phi-1}$ & $0<z<1/a$
 &$\mbox{Gen. Gamma}$ \\ \hline
\end{tabular}
\end{center}
\caption{The form of directing L\'evy intensity in a CoRM which leads to particular marginal processes.}\label{t:derived1}
\end{table}
The results are surprising. A gamma marginal process arises when the directing L\'evy process is a Beta process and a $\sigma$-stable marginal process arises when the directing L\'evy process is also a $\sigma$-stable process. Generalized gamma marginal processes lead to a directing L\'evy process which is a generalization of the Beta process (with a power of $z$ which is less than 1) and re-scaled to the interval $(0,1/a)$. In fact, if we 
use a gamma score distribution with shape $\phi$ and mean $a\phi$ which has density
\begin{equation}\label{Jumps}
f(x)=\frac{1}{a^{\phi}\Gamma(\phi)}x^{\phi-1}\exp\{-x/a\},
\end{equation}
 the directing L\'evy intensity is a stable Beta \cite{tehgor} of the form
\[
\nu^{\star}(z)=\frac{a^{\sigma+1}\sigma}{\phi} \frac{\Gamma(\phi+1)}{\Gamma(\phi+\sigma)\Gamma(1-\sigma)}
z^{-\sigma-1}(1-z)^{\sigma+\phi-1},\qquad 0<z<1.
\]

\begin{rmk} 
Several authors have previously considered hierarchical models where $\tilde\mu_1,\dots,\tilde\mu_d$ followed i.i.d. CRM (or NRMI) processes whose centring measure are given a CRM  (or NRMI) prior.  This construction induces correlation between $\tilde\mu_1,\dots,\tilde\mu_d$ and the hierarchical Dirichlet process is a popular example but we will concentrate on a hierarchical Gamma process \cite[see {\it e.g.}][]{Palla}. In this case, $\tilde\mu_1,\dots,\tilde\mu_d$ follow independent Gamma processes with centring measure $\alpha$ which also follows a Gamma process. This implies that we can write
\[
\alpha = \sum_{i=1}^{\infty} s_i\delta_{\theta_i}
\]
and we can write
\begin{equation}
\tilde\mu_j=\sum_{i=1}^{\infty} J_{j,i}\delta_{\theta_i}
\label{eqn_J}
\end{equation}
where $J_{j,i}\sim\mbox{Ga}(s_i)$. This can be represented as a CoRM process where $\alpha$ is the directing L\'evy process and the score distribution is $\prod_{j=1}^d \mbox{Ga}(s_i)$ where $s_i$ controls the shape of the  conditional distribution of $J_{j,i}$. This contrasts with the processes considered in this section with independent gamma scores which multiply the jumps in the directing L\'evy process and lead to a marginal gamma process for $\tilde\mu_j$ (unlike the hierarchical model). These processes can be written in the form of (\ref{eqn_J}) with $J_{j,i}$ having a gamma distribution with shape $\phi$ and mean $\phi s_i$, and $\alpha$ chosen to follow a beta process. 
\end{rmk}

\begin{rmk}
This paper is focused on Gamma scores but the class of CoRMs is very wide and other choices can be considered. For instance, if $Beta(\alpha,1)$ scores are selected, i.e. 
$$f(x)=\alpha x^{\alpha-1}\qquad \alpha>0,\quad 0<x<1$$
then it is possible to introduce a multivariate version of the Beta process. Let $\nu(s)=\theta s^{-1}(1-s)^{\theta-1}$, $0<s<1$, i.e. the L\'evy intensity of the jumps of a Beta process,
 then $\nu^{\star}(z)$ is the solution of the integral equation
\[
\nu(s)=\int_s^1 f(s/z)s^{-1}\nu^{\star}(z)\,dz,\quad 0<s<1
\]
A simple application of the fundamental Theorem of Calculus leads to
\[
\nu^{\star}(z)=
\theta z^{-1}(1-z)^{\theta-1}+\frac{\theta(\theta-1)}{\alpha}(1-z)^{\theta-2}
\]
which is the sum of $\nu(\cdot)$, the L\'evy intensity of the original Beta process, and a compound Poisson process (if $\theta>1$) with intensity $\theta/\alpha$ and jump distribution $Beta(1,\theta-1)$. This is well-defined if $\theta>1$. 
\end{rmk}
It is interesting to derive the resulting multivariate L\'evy intensities which can be compared with similar results in \cite{LL}, \cite{LLS} and \cite{ZL}. 

\begin{thm}\label{Gamma-GammaLevy}
Consider a CoRM process with independent $\mbox{Ga}(\phi, 1)$ distributed scores. If the CoRM process has gamma process  marginals  then
\begin{equation}
\rho_d(s_1,\dots,s_d)=\frac{(\prod_{j=1}^ds_j)^{\phi-1}}{[\Gamma(\phi)]^{d-1}}|\bm{s}|^{-\frac{d\phi+1}{2}}e^{-\frac{|\bm{s}|}{2}}W_{\frac{(d-2)\phi+1}{2},-\frac{d\phi}{2}}(|\bm{s}|)
\end{equation}
where $|\bm{s}|=s_1+\dots+s_d$ and $W$ is the  Whittaker function. If  the CoRM process has $\sigma$-stable process marginals  then 
\begin{equation}
\rho_d(s_1,\dots,s_d)=\frac{(\prod_{j=1}^ds_j)^{\phi-1}}{[\Gamma(\phi)]^{d-1}}\frac{\Gamma(\sigma+d\phi)}{\Gamma(\sigma)\Gamma(1-\sigma)}|\bm{s}|^{-\sigma-d\phi}.
\end{equation}
\end{thm}
The result is proved in the appendix with the following corollary.
\begin{coro}\label{Gamma-GammaCoro}
Consider a CoRM process with independent exponentially distributed scores. If the CoRM has gamma process marginals we recover the multivariate L\'evy intensity of \cite{LLS},
$$\rho_d(s_1,\dots,s_d)=\sum_{j=0}^{d-1}\frac{(d-1)!}{(d-1-j)!}|\bm{s}|^{-j-1}e^{-|\bm{s}|}.$$
Otherwise, if $\sigma$-stable marginals are considered then we recover the multivariate vector introduced in \cite{LL} and \cite{ZL},
$$\rho_d(s_1,\dots,s_d)=\frac{(\sigma)_d}{\Gamma(1-\sigma)}|\bm{s}|^{-\sigma-d}.$$
\end{coro}
Alternatively, we can specify $\nu^{\star}$ and derive $\nu$. The forms for some particular processes are shown in Table~\ref{t:derived2}
\begin{table}[h!]
\begin{center}
\begin{tabular}{rl}\hline
$\nu(s)$ & Directing L\'evy process\\\hline
$\frac{\Gamma(\theta+1)}{\Gamma(\phi)}s^{-1}
\exp\{-s\}U(\theta-\phi,1-\phi, s)$ & Beta \\
$2\frac{1}{\Gamma(\phi)}\frac{\sigma}{\Gamma(1-\sigma)}s^{(\phi-\sigma)/2-1}
a^{(\sigma+\phi)/2}
K_{\sigma+\phi}\left(2\sqrt{a s}\right)$
  & Gen. Gamma \\\hline
\end{tabular}
\end{center}
\caption{The L\'evy intensity of the marginal process in a CoRM with different directing L\'evy processes.}
\label{t:derived2}
\end{table}
where $U$ is the confluent hypergeometric function of the second kind and $K$ is the modified Bessel function of the second kind.

\begin{rmk}

There are several special cases if $\nu^{\star}$ is the L\'evy intensity of a Beta process. Firstly, $U(\theta-\phi,1-\phi,s)=1$ if $\theta=\phi$ and $\nu$ is the L\'evy intensity of a gamma process. 
If $\phi = 2\theta-1$, $$U(\theta-\phi, 1-\phi, s)=\pi^{-1/2}\exp\{s/2\}s^{1/2-\theta+\phi}K_{\theta-1/2}(s/2).$$
When $\theta=1$, $U(1-\phi,1-\phi, s) = \exp\{s\}\int_{s}^{\infty} u^{-(1-\phi)}\exp\{-u\}\,du$. The limits as $s\rightarrow 0$ are 
\[
U(\theta-\phi,1-\phi,s)\rightarrow\left\{
\begin{array}{lc}
\Gamma(\phi)/\Gamma(\theta)+O(|s|^{\phi}) & 0<\phi<1\\
1/\Gamma(1+\theta-\phi)+O(|s\log s|) & \phi=1\\
\Gamma(\phi)/\Gamma(\theta)+O(|s|) & \phi>1
\end{array}
\right.
\]
Therefore, these processes have a L\'evy intensity similar to the L\'evy intensity of the gamma process 
 close to zero for any choice of $\phi$ and $\theta$. The tails of the L\'evy intensity are exponential. Therefore, the process has similar
properties to the gamma process.
\end{rmk}

\begin{rmk}
The generalized gamma process  contains some special cases and the L\'evy intensity of the marginal process for these process are shown in Table~\ref{t:derived3}.
\begin{table}[h!]
\begin{center}
\begin{tabular}{rl}\hline
$\nu(s)$ & Directing L\'evy process\\\hline
$2\frac{1}{\Gamma(\phi)}s^{\phi/2-1}
K_{\phi}\left(2\sqrt{s}\right)$ & Gamma Process\\
$\frac{\Gamma(\phi+\sigma)}{\Gamma(\phi)}\frac{\sigma}{\Gamma(1-\sigma)}s^{-1-\sigma}$ &  $\sigma$-stable Process.\\\hline
\end{tabular}
\end{center}
\caption{The L\'evy intensity of the marginal process in a CoRM with different directing L\'evy processes.}
\label{t:derived3}
\end{table}
With a generalized gamma directing L\'evy process,
It is straightforward to show that 
\[
\nu(s)\approx
\sigma\frac{\Gamma(\sigma+\phi)}{\Gamma(\phi)\Gamma(1-\sigma)}
 s^{-\sigma-1}
\]
for small $s$.
Therefore, the L\'evy intensity close to zero is similar to the L\'evy intensity of $\sigma$-stable process with parameter $\sigma$. For large $s$, we have 
\[
\nu(s)s\propto
\sqrt{\pi}\frac{1}{\Gamma(\phi)}\frac{\sigma}{\Gamma(1-\sigma)}(a s)^{(\phi+\sigma)/2-1/4}
s^{-1-\sigma}
\exp\{-2\sqrt{a}s^{1/2}\}.
\]
Therefore, the tails will decays like $\exp\{-s^{1/2}\}$.
\end{rmk}

The next Theorems will provide an expression of the Laplace exponent when the scores are gamma distributed with $\phi\geq 1$ such that $\phi\in\mathbb{N}$. We want to stress the importance of the the Laplace transform in the Bayesian nonparametric setting. Indeed, it is the basis to prove theoretical results of the prior of interest. For instance, \cite{LL}, \cite{LLS} and \cite{ZL} used the Laplace Transform to derive some distributional properties such as correlation, partition structure and mixed moments. Additionally, we will see that the Laplace transform plays a role in the novel sampler proposed in this paper.      

\begin{thm}\label{LapFinal}
Consider a CoRM process with  independent $\mbox{Ga}(\phi, 1)$ distributed scores.
Suppose $\phi\geq 1$ such that $\phi\in\mathbb{N}$. Let $\bm{\lambda}\in(\mathbb{R}^+)^d$ be a vector such that it consists of $l\le d$ distinct values denoted as $\bm{\tilde{\lambda}}=(\tilde{\lambda}_1,\ldots,\tilde{\lambda}_l)$ with respective multiplicities $\bm{n}=(n_1,\dots,n_l)$. 
Then
\begin{equation*}
\label{eq:casogenerale}
\psi_{\rho,d}(\bm{\lambda})
=\psi_{\rho,d}(\bm{\tilde{\lambda}},\bm{n})=\frac{[\Gamma(\phi)]^l}{\prod_{i=1}^l [\tilde{\lambda}_i^{\phi-1}\Gamma(n_i\phi)]}\left(\prod_{i=1}^l\frac{\partial^{(n_i-1)\phi}}{\partial^{(n_i-1)\phi}\tilde{\lambda}_i}\right)\left(\Upsilon_l^{\phi}(\bm{\tilde{\lambda}})\prod_{i=1}^l \tilde{\lambda}_i^{n_i\phi-1}\right),
\end{equation*}
where 
$$\Upsilon^{\phi}_l(\bm{\tilde{\lambda}})=\int \left(1-\prod_{i=1}^{l} \frac{1}{(1+z\tilde{\lambda}_i)^{\phi}}\right)\nu^{\star}(z)dz.$$
\end{thm}

The proof of the previous Theorem is based on the result provided in Theorem \ref{mgf_thm} since the moment generating of a Gamma distribution exists and it is explicit. 

To compute the expression of $\Upsilon_l^{\phi}(\bm{\tilde{\lambda}})$ we need to define the following set
$$A_{\phi,j}=\{\bm{k}\in\{1,\ldots,\phi\}^j:\: |\bm{k}|=\phi\}\qquad \phi\geq j.$$

\begin{thm} 
Consider a CoRM process with independent $\mbox{Ga}(\phi, 1)$ distributed scores.
Suppose $\phi\geq 1$ such that $\phi\in\mathbb{N}$. 
Let $\Lambda(\bm{\tilde{\lambda}},\bm{z})=(1-\sum_{h=1}^{j-1}z_h)\tilde{\lambda}_{i_j}+\sum_{h=1}^{j-1} \tilde{\lambda}_{i_h}z_h$ be a function defined on the (j-1)-dimensional simplex $$\Delta_{j-1}=\{\mathbf{z}\in (0,1)^{j-1}:z_1+\cdots+z_{j-1}<1\}$$ with the convention that $\Delta_0=[0,1]$ . Let  
$$a_i(\bm{\tilde{\lambda}})=\frac{\tilde{\lambda}_i^{l-1}}{\prod_{\substack{j=1\\j\neq i}}^l (\tilde{\lambda}_i-\tilde{\lambda}_j)}\qquad i=1,\dots,l.$$ 
then
\begin{equation*}
\Upsilon_l^{\phi}(\bm{\tilde{\lambda}})=\left\lbrace\begin{array}{ll}
\phi!\displaystyle\sum_{j=1}^{\phi}\sum_{\bm{k}\in A_{\phi,j}}\sum_{0<i_1<i_2<\cdots<i_j\leq l} \frac{a_{i_1}^{k_1}(\bm{\tilde{\lambda}})\cdots a_{i_j}^{k_j}(\bm{\tilde{\lambda}})}{k_1!\dots k_j!}C(i_1,\dots,i_j;\bm{k};\bm{\tilde{\lambda}})&\mbox{ if } l>1\\
\psi(\lambda_1)&\mbox{ if } l=1\end{array}\right.
\label{eq:phi}
\end{equation*}
where
$$
C(i_1,\dots,i_j;\bm{k};\bm{\tilde{\lambda}})=
\Gamma(\phi)\int_{\Delta_{j-1}} \left((1-\sum_{h=1}^{j-1}z_h)^{k_j}\prod_{h=1}^{j-1} \frac{z_h^{k_h-1}}{\Gamma(k_h)}\right)\psi\left(\Lambda(\bm{\tilde{\lambda}},\bm{z})\right)d\bm{z}
$$
For the above integral we assume the usual convention that $\sum_{i}^j=0$ and $\prod_{i}^j=1$ whenever $i>j$.  
\end{thm}

In the following Corollary, the expression of the Laplace exponent is recovered for the special case of a CoRM with independent exponentially distributed scores. 
 
\begin{coro}
Consider a CoRM process with independent exponentially distributed scores. It follows   that
\begin{equation*}
\psi_{\rho,d}(\bm{\lambda})
=\psi_{\rho,d}(\bm{\tilde{\lambda}},\bm{n})=\left(\prod_{i=1}^l\frac{1}{\Gamma(n_i)}\frac{\partial^{(n_i-1)}}{\partial^{(n_i-1)}\tilde{\lambda}_i}\right)\left(\Upsilon_I(\bm{\tilde{\lambda}})\prod_{i=1}^l \tilde{\lambda}_i^{(n_i-1)}\right),
\end{equation*}

\begin{equation*}
\Upsilon_l(\bm{\tilde{\lambda}})=\left\lbrace\begin{array}{ll}
\sum_{i=1}^l a_{i}(\bm{\tilde{\lambda}})\psi(\lambda_i)&\mbox{ if } l>1\\
\psi(\lambda_1)&\mbox{ if } l=1\end{array}\right..
\label{eq:Expocase}
\end{equation*}
\end{coro}
The proof of the corollary is omitted since it is a direct application of the results of the previous Theorems. Note that, if the vector has Gamma process marginals, {\it i.e.} $\psi(\lambda_i)=\log(1+\lambda_i)$, then we recover the results in \cite{LLS}. If the vector has $\sigma$-stable process marginals, {\it i.e.} $\psi(\lambda_i)=\lambda_i^{\sigma}$, then we recover the result in \cite{LL} and \cite{ZL}. \\

Finally, we close the section with some results about the dependence structure of CoRM processes. A useful description of the dependence of a vector of CRMs is given by the L\'evy copula. A L\'evy Copula is a mathematical tool that allows the construction of multivariate L\'evy intensities with fixed marginals, see appendix. 
The following Theorem displays the underlying L\'evy Copula of a compound random measure. 
\begin{thm}
Let $\rho_d$ be the compound random measure defined in \eqref{compound} and let $F$ be the the distribution function of $f$. The underlying L\'evy Copula of the compound random measure is 
$$C(s_1,\dots,s_d)=\int \nu^{\star}(z)\prod_{j=1}^d (1-F(z^{-1}U^{-1}(s_j)))dz$$
where $U^{-1}$ is the inverse of the tail integral function $U(x):=\int_x^\infty \nu(s)\,\ddr s$. 
\end{thm}

Furthermore, it is possible to prove a result similar to Proposition 5 in \cite{LLS}. This result gives a close formula for the mixed moments of two dimensions of a CoRM process. 
The result is expressed in terms of an ordering on sets $\bm{0} \prec\bm{s}_1\prec\,\cdots\,\prec\bm{s}_j$ which is defined in \cite{constsav}.

\begin{thm}
Consider a CoRM process with an independent $\mbox{Ga}(\phi, 1)$ distributed scores. Let $\bm{q}=(q_1,\dots,q_d)$ and let $p_j(\bm{q},k)$ be the set of vectors
$(\bm{\eta},\bm{s}_{1},\ldots,\bm{s}_j)$ such that the coordinates
of $\bm{\eta}=(\eta_1,\ldots,\eta_j)$ are positive and such
that $\sum_{i=1}^j \eta_i=k$. Moreover, 
$\bm{s}_i=(s_{1,i},\dots, s_{d,i})$ are vectors such that
$\bm{0} \prec\bm{s}_1\prec\,\cdots\,\prec\bm{s}_j$ and $\sum_{i=1}^j \eta_i(s_{1,i}+\cdots+s_{d,i})=k=q_1+\dots+q_d$. 
 Then,
\begin{equation*}
\begin{split}
\E\left[\prod_{i=1}^d\{\tilde\mu_i(A)\}^{q_i}\right]&=q_1!\cdots q_d!\: \sum_{k=1}^{\mid \bm{q}\mid}[\alpha(A)]^k\:\times\:\\
&\times  \sum_{j=1}^{\mid \bm{q}\mid} \:\sum_{p_j(\bm{q},k)}\:\prod_{i=1}^j \frac{1}{\eta_i!} \left[\left(\prod_{l=1}^d \frac{(\phi)_{s_{l,i}}}{s_{l,i}!}\right)\int z^{{s_{1,i}+\cdots + s_{d,i}}}\nu^{\star}(z)dz\right]^{\eta_i}
\end{split}
\end{equation*}
where $\mid \bm{q}\mid=q_1+\cdots+q_d$.
\end{thm}

\begin{rmk} For instance, suppose that the CoRM process has generalized gamma process marginals. Then,
\begin{equation*}
\begin{split}
\int z^{s_{1,i}+\cdots + s_{d,i}}\nu^{\star}(z)dz=\frac{\sigma a^{\sigma-(s_{1,i}+\cdots + s_{d,i})}}{\Gamma(1-\sigma)}B(k-\sigma-1,\sigma+\phi).
\end{split}
\end{equation*}
\end{rmk}

\section{Normalized Compound Random Measures}

Vectors of correlated random probability measures can be defined by normalizing each dimension of a CoRM process. This will be called a Normalized Compound Random Measure (NCoRM) and is  defined by a score distribution, a directing L\'evy process and a centring measure of the CoRM. The results derived in 
Table \ref{t:derived1}
 can be used to define a NCoRM with a particular marginal process. For example, an NCoRM with Dirichlet process marginals arises by normalizing each dimension of a CoRM with gamma process marginals. 

In specifying an NCoRM prior, it is useful to have a method of choosing the parameters of the score distribution to give a particular level of dependence. We describe two possible methods.
 It is possible to compute the covariance of a two dimensions of an NCoRM process. Indeed, following \cite{LLS},
\begin{equation}
  \begin{split}
    \label{eq:cross_moment_simple}
    \mbox{Cov}\left[\tilde p_1(A),\,\tilde p_2(B)\right]
    &=  \left\{\alpha(A\cap B)-
      \frac{\alpha(A)\alpha(B)}{\alpha(\mathbb{X})}\right\} \\[7pt]
    &\qquad \times\:\int_{(\R^+)^2}g_\rho(1,1;\lambda_1,\lambda_2)\,
    \edr^{-\alpha(\mathbb{X})\psi_\rho(\lambda_1,\lambda_2)} \:\ddr \lambda_1\,\ddr \lambda_2
  \end{split}
\end{equation}
where $g_{\rho}$ is the function introduced in Equation \eqref{grho}.
This result can be  used to specify any parameters of the score distribution (or a prior for those parameters). Alternatively, if the scores are independent, the ratio of the same jump heights in the $i$-th and $j$-th dimension has the same distribution as  the ratio of two independent random variables following the score distribution. For example, if the scores are independent and follow a gamma  distribution with shape $\phi$  is chosen, this ratio follows an $F$-distribution with $\phi$ and $\phi$ degrees of freedom.

\subsection{Links to other processes}

Corollary~\ref{Gamma-GammaCoro}  shows how the priors described in 
 \cite{LL}, \cite{LLS} and \cite{ZL} can be expressed in the CoRM framework. The  CNMRI process \citep{GKS13}\citep[see also][]{NLP14a, Chen13}  can also be expressed in the NCoRM framework. The CNMRI prior express the random measure $\tilde\mu_g$ as 
\[
\tilde\mu_j=\sum_{k=1}^q D_{jk} \tilde\mu^{\star}_k
\]
where $D$ is a $(d\times q)$-dimensional selection matrix (with elements either equal to 0 or 1) and $\tilde\mu^{\star}_1,\dots,\tilde\mu^{\star}_q$ are independent CRMs where $\tilde\mu^{\star}_k$ has L\'evy intensity $M_k\nu^{\star}(ds)\bar\alpha(dx)$ for a probability measure $\bar\alpha$. 
A CNRMI process can be represented by a vector of CoRMs with score probability mass function
\[
g(s_1=D_{1i}z,\dots,s_d=D_{di}z\vert z)=\frac{M_i}{\sum_{k=1}^q M_k},
\]
 directing L\'evy intensity $\nu^{\star}$ and centring measure $\bar\alpha\sum_{k=1}^q M_k$. A CoRM process with independent scores  can be used to construct a sub-class of CNRMI processes. A CoRM has a score distribution of the form $f(s)=\pi\delta_{s=1} + (1-\pi)\delta_{s=0}$, directing L\'evy intensity $\nu^{\star}(ds)$ and centring measure $M\bar\alpha$
  is identical to an unnormalized CNRMI process with $q=2^d$, a $D$ whose rows are the binary expansion of $\{0,1,\dots,2^d-1\}$ and $M_k=M\prod_{l=1}^d  \pi^{D_{kl}}(1-\pi)^{1-D_{kl}}$. 
A more general class  of unnormalized CNRMI processes with
$M_k=M\prod_{l=1}^d  \pi_l^{D_{kl}}(1-\pi_l)^{1-D_{kl}}$  which corresponds  to a vector of CRMs such that
\begin{equation}\label{compound2}
\rho_d(ds_1,\dots,ds_d)=\int  z^{-d}\prod_{j=1}^d  f_j(s_j/z)\,ds_1\cdots ds_d
\,\nu^{\star}(dz)
\end{equation}
where $f_j(m) = \pi_j\delta_{m=1}+(1-\pi_j)\delta_{m=0}$.

\subsection{Computational Methods}\label{sec:comp_method}

We describe methods for fitting  a nonparametric mixture model where the mixing measure is given  a NCoRM prior. We assume that the data can be divided into $d$ groups and $y_{j,1},\dots,y_{j,n_j}$ are the observations in the $j$-th group. The data are modelled as
\[
y_{j,i}\stackrel{ind.}{\sim} k(y_{j,i}\vert \zeta_{j,i}),
\quad
\zeta_{j,i}\sim\tilde{p}_j,\quad  i=1,2,\dots,n_j,\quad j=1,\dots,d
\]
where 
$k(y\vert \theta)$ is a probability density function for $y$ with parameter $\theta$ 
and 
$\tilde{p}_1,\dots,\tilde{p}_d$ are given an NCoRM prior.
Using the notation of (\ref{jump_mug}), we write
\[
\tilde{p}_j = \frac{\tilde\mu_j}{\tilde\mu_j(\mathbb{X})}
=\frac{\sum_{k=1}^{\infty} m_{j,k} \,J_k\,  \delta_{\theta_k}}{\sum_{k=1}^{\infty} m_{j,k}\,J_k}.
\]

Direct simulation from the posterior distribution is impossible since there are an infinite number of parameters. Several MCMC methods have been introduced which circumvent this problem in the class of normalized random measure mixtures.
\cite{Favaro2013} describe an auxiliary variable method which involves integrating out the unnormalized random measure whereas
\cite{GW2011} introduce a slice sampling method. We consider extending both methods to NCoRM mixtures.

  We use the notation $m=(m_{j,k})$, $J=(J_1,J_2,\dots)$ and $\theta=(\theta_1,\theta_2,\dots)$. The
 posterior distribution can be expressed in a suitable form for MCMC by introducing latent variables.
 Firstly, latent allocation variables $c=(c_{j,i})$ (for which $\zeta_{j,i}=\theta_{c_{j,i}}$) are introduced to give
 \begin{align}
 p(y,c\vert m, J, \theta)&=
 \prod_{j=1}^d \prod_{i=1}^{n_j} 
\left[ k\left(y_{j,i}\vert \theta_{c_{j,i}} \right)
 \frac{m_{j,c_{j,i}} \,J_{c_{j,i}}}{\sum_{k=1}^{\infty} m_{j,k}\,J_k}\right]\nonumber\\
&= \prod_{j=1}^d \frac{\prod_{i=1}^{n_j} 
 k\left(y_{j,i}\vert \theta_{c_{j,i}} \right)
  m_{j,c_{j,i}} \,J_{c_{j,i}}}{\left(\sum_{k=1}^{\infty} m_{j,k}\,J_k\right)^{n_j}}.
  \label{post1}
 \end{align}
 Secondly, latent variables $v=(v_1,\dots,v_d)$ are introduced to define
\begin{align*}
 p(y,c,v\vert m, J, \theta)
 =& \prod_{j=1}^d \left[\prod_{i=1}^{n_j} 
 k\left(y_{j,i}\vert \theta_{c_{j,i}} \right)
  m_{j,c_{j,i}} \,J_{c_{j,i}}\right]
\prod_{j=1}^d\left[\frac{1}{\Gamma(n_j)} v_j^{n_j-1}\right]\\
&\times 
\exp\left\{-\sum_{j=1}^d v_j  \sum_{k=1}^{\infty} m_{j,k}\,J_k\right\}.
 \end{align*}
 Integrating over $v$ (using the identity $\frac{1}{\Gamma(n)} v^{n-1}\exp\{-vx\}=x^{-n}$) gives the expression in (\ref{post1}).

\subsubsection{Marginal method}

The \cite{Favaro2013} approach relies on an analytical form for $p(y,v, c)$ which is available for the NRMI mixtures using results of \cite{JLP09}. Suppose $\{c_{j,i}\}$ takes $K$ distinct values, that $a_{j,k}$ is the number of observations in the $j$-th group allocated to the $k$-th distinct value and define $a_k=(a_{k,1},\dots,a_{k,d})$. Extending the results of 
\cite{JLP09} and \cite{Favaro2013}  to vectors of normalized random measures (as in Section 2.1) leads to 
\[
p(y,v, c)=
\prod_{i=1}^d\frac{1}{\Gamma(n_i)}v_i^{n_i-1} 
\exp\{-\psi_{\rho,d}(v)\}\prod_{k=1}^K \kappa_{a_k}(v) \prod_{k=1}^K g(
\{y_{j,i}\vert c_{j,i}=k\})
 \]
where 
\[
\psi_{\rho,d}(v) = \int \left(1 - \exp\left\{-\sum_{i=1}^d v_i s_i\right\}\right) \rho_d(ds_1,\dots,ds_d),
\]
\[
\kappa_a(v)=\int \prod_{j=1}^d s_j^{a_j}
\exp\left\{-\sum_{i=1}^d v_i s_i\right\}\rho_d(ds_1,\dots,ds_d)
\]
and
\[
g(y)=\int \prod k(y_{j,i}\vert \theta) \alpha(d\theta).
\]
If the vector of the normalized random measures is chosen to be an NCoRM with independent gamma scores then
\begin{align*}
\kappa_a(v)&=\int \prod_{j=1}^d s_j^{a_j}
\exp\left\{-\sum_{i=1}^d v_i s_i\right\}
z^{-d}\prod_{j=1}^d f(s_j/z)\,ds_1\dots ds_d\,
\nu^{\star}(dz)\\
&=
\int z^{\sum_{j=1}^d a_j}
 \prod_{j=1}^d \int \left[s_j^{a_{j}}
\exp\left\{- v_j z s_j \right\}
 f(s_j)\, ds_j\right]\nu^{\star}(dz)
\\
&=
\int z^{\sum_{j=1}^d a_j}
 \prod_{j=1}^d \tau_{a_j}(z,v_j)\nu^{\star}(dz)
\end{align*}
where
\[
\tau_a(z, v)=\int s^{a}\exp\left\{- v z s\right\}
 f(s)\, ds.
\]
and Theorem 3.1 provides the expression
\begin{align*}
\psi_{\rho,d}(v) &
=\int \left(1-\prod_{j=1}^{d} M_z^f(-s_j)\right)\nu^{\star}(z)dz.
\end{align*}
If $f$ is chosen to be a gamma distribution with shape parameter $\phi$,
\begin{align*}
\tau_a(z,v)
&=\int s^a\exp\left\{- v z s\right\} f(s)\, ds
 =\frac{\Gamma(a+\phi)}{\Gamma(\phi)}(1+vz)^{-a-\phi}.
\end{align*}

Two algorithms can be defined. One is suitable for conjugate mixtures where $g(y)$ can be calculated analytically and a second algorithm is suitable for non-conjugate mixtures where $g(y)$ cannot be calculated analytically.

In the case of a conjugate mixture model, the steps of the algorithm are

\subsubsection*{Updating $c_{j,i}$}

Let $C^{-(j,i)}_k=\{y_{l,m}\vert c_{l,m}=k, (l,m)\neq (j,i)\}$ and $K^{-(j,i)}$ be the number of distinct values of $\{c_{l,m}\vert (l,m)\neq (j,i)\}$. The parameter $c_{j,i}$ is updated from the discrete distribution
\[
p(c_{j,i}=k)\propto
\left\{
\begin{array}{ll}
\frac{\kappa_{a_k+r}(v)g\left(C_k^{-(j,i)}\cup \{y_{j,i}\}\right)}{\kappa_{a_k}(v)g\left(C_k^{-(j,i)}\right)}   & 1\leq k\leq K^{-(j,i)}\\
\kappa_r(v)g(y_{j,i}) & k = K^{-(j,i)} + 1
\end{array}
\right.
\]
where $r$ is a $d$-dimensional vector with $r_m=1$ if $m=j$ and $r_m=0$ otherwise.
For independent $\mbox{Ga}(\phi, 1)$ scores, 
\[
\frac{\kappa_{a_k+r}(v)}{\kappa_{a_k}(v)}
= (a_{j,k}+\phi)
\frac{\int z^{\sum_{m=1}^d a_{m,k}+1}
(1+v_j z)^{-a_{j,k}-1-\phi} \prod_{m=1;m\neq j}^d (1+v_m z)^{-a_{m,k}-\phi}\nu^{\star}(z)\,dz}
 {\int z^{\sum_{m=1}^d a_{m,k}}
 \prod_{m=1}^d (1+v_m z)^{-a_{m,k}-\phi}\nu^{\star}(z)\,dz}
\]
and 
\[
\kappa_r(v)=
  \phi
\int z
(1+v_j z)^{-1-\phi} \prod_{m=1; m\neq j}^d (1+v_m z)^{-\phi}\nu^{\star}(z)\,dz.
\]

\subsubsection*{Updating $v_j$}

The full conditional distribution of $v_j$ is  proportional to
\[
 v_j^{n_j-1} 
\exp\{-\psi_{\rho,d}(v)\}\prod_{k=1}^K \kappa_{a_k}(v).
 \]
This parameter can be updated using an adaptive Metropolis-Hastings random walk \citep{atros}.

\subsubsection*{Updating parameters of $f$}

The full conditional distribution of the parameters of $f$ is  proportional to
\[
\exp\{-\psi_{\rho,d}(v)\}\prod_{k=1}^K \kappa_{a_k}(v).
 \]
This parameter can be updated using an adaptive Metropolis-Hastings random walk \citep{atros}.

In the case of non-conjugate mixtures, \cite{Favaro2013} define an auxiliary variable method which introduces the distinct values $\theta_1,\dots,\theta_K$ into the sampler and $M$ potential distinct values for empty clusters $\theta'_1,\dots,\theta'_M$.

\subsubsection*{Updating $c_{j,i}$}

A set of values $\theta_1,\dots,\theta_M$ is formed. If $c_{j,i}$ is a singleton ({\it i.e.} $c_{j,i}\neq c_{k,m}$ for $(j,i)\neq (k,m)$), set $\theta'_1=\theta_{c_{j,i}}$ and sample $\theta'_j\sim \alpha/\alpha(\mathbb{X})$ for $j=2,\dots,M$. Otherwise, sample 
$\theta'_j\sim \alpha/\alpha(\mathbb{X})$ for $j=1,\dots,M$. The full conditional distribution of $c_{j,i}$ is
\[
p(c_{j,i}=k)\propto
\left\{
\begin{array}{ll}
\frac{\kappa_{a_k+r}(v)}{\kappa_{a_k}(v))}k(y_{j,i}\vert \theta_k)   & 1\leq k\leq K\\
\frac{\alpha(\mathbb{X})}{M}\kappa_r(v) k\left(y_{j,i}\vert \theta'_{k-K^{-(j,i)}}\right) & k = K^{-(j,i)} + 1,\dots,K^{-(j,i)}+M.
\end{array}
\right.
\]

\subsubsection*{Updating $\theta_k$}

The full conditional density of $\theta_k$ is proportional to
\[
\alpha(\theta_k)\prod_{\{(j,i)\vert c_{j,i}=k\}}k(y_{j,i}\vert \theta_k).
\]

The full conditional distributions of $v_j$ and any parameters of $f$ are unchanged from algorithm for conjugate mixture models.

\subsubsection{Slice sampling method}

We introduce $u=(u_{j,i})$ and define
\begin{align*}
 p(y,c,v,u\vert m, J, \theta)
 =& \prod_{j=1}^d \left[\prod_{i=1}^{n_j} 
 k\left(y_{j,i}\vert \theta_{c_{j,i}} \right)
  m_{j,c_{j,i}} \,\mbox{I}(u_{j,i}<J_{c_{j,i}})\right]
\prod_{j=1}^d\left[\frac{1}{\Gamma(n_j)} v_j^{n_j-1}\right]\\
&\times \exp\left\{-\sum_{j=1}^d v_j  \sum_{k=1}^{\infty} m_{j,k}\,J_k\right\}.
 \end{align*} 
 Integrating over $u$ and $v$ gives the expression in (\ref{post1}). A similar form is derived in \cite{GW2011}.
This form of the likelihood is still not suitable for MCMC since it involves all jumps. To avoid this, 
we define $L=\min_{i=1,\dots,n_j; j=1,\dots,d}\left\{u_{j,i}\right\}$ and divide the jumps into two disjoints sets:  $A^{\dagger}=\{(J^{\dagger}_k, m^{\dagger}_{1,k},\dots,m^{\dagger}_{d,k})\vert J^{\dagger}_k>L\}$ 
 and $A^{\star}=\{(J^{\star}_k, m^{\star}_{1,k},\dots,m^{\star}_{d,k})\vert J^{\star}_k\leq L\}$. The set $A^{\dagger}$ has a finite number of elements which is denoted $K$ and $A^{\star}$ has an infinite number of elements.
Integrating over $A^{\star}$ leads to posterior which is suitable for MCMC and has the form
 \begin{align}
&\prod_{j=1}^d 
\left[\prod_{i=1}^{n_j}
k\left(y_{j,i}\vert \theta_{c_{j,i}}\right) m^{\dagger}_{j,c_{j,i}}\,\I\left(u_{j,i}< J_{c_{j,i}}\right)
\right]
\prod_{j=1}^d 
\left[\frac{1}{\Gamma(n_j)}v_j^{n_j-1}\right]
\nonumber\\
&\times\exp\left\{-\sum_{j=1}^d v_j\sum_{k=1}^K m^{\dagger}_{j,k}J^{\dagger}_k\right\}\E\left[
\exp\left\{-\sum_{j=1}^d v_j\sum_{k=1}^{\infty} m^{\star}_{j,k}J^{\star}_k\right\}
\right].
\label{post_MCMC}
\end{align}
An MCMC scheme using this form of likelihood leads to  a random truncation of the NCoRM process  at each iteration  but does not introduce a truncation error since integrating over the latent variables leads to the correct marginal posterior. 
 
The expectation in (\ref{post_MCMC}) can be expressed in terms of a univariate integral using a variation on
Theorem \ref{mgf_thm} giving
\[
-\log \E\left[
\exp\left\{-\sum_{j=1}^d v_j\sum_{k=1}^{\infty} m^{\star}_{j,k}J^{\star}_k\right\}
\right]=
\int_0^L \left(1-\prod_{j=1}^{d} M_z^f(-v_j)\right)\nu^{\star}(z)\,dz.
\]
 The full conditional distributions and a general discussion of methods for updating parameters are given below. Details of the implementation for specific processes are given in the appendix.

\subsubsection*{Updating $v_1,\dots,v_d$}

The updating of $v_1,\dots,v_d$ uses a variation on the interweaving approach of 
\cite{yumeng11}, which leads to better mixing than the standard full conditional distribution for $v_j$.
The parameter  $v_j$ is updated in the following way. Firstly, we re-parameterize to $\tilde{m}^{\dagger}_{j,k}=v_j m^{\dagger}_{j,k}$ and update $v_j$ from the full conditional density (conditioning on $\tilde{m}^{\dagger}_{j,k}$ rather than $m^{\dagger}_{j,k}$) which is proportional to
\[
v_j^{-(K+1)}f\left(\frac{\tilde{m}^{\dagger}_{j,k}}{v_j}\right)
\E\left[
\exp\left\{-v_j\sum_{k=1}^K m^{\star}_{j,k} J^{\star}_k\right\}
\right].
\]
Secondly, we re-parameterized to $m^{\dagger}_{j,k}=\tilde{m}^{\dagger}_{j,k}/v_j$
and update $v_j$ from the full conditional density  proportional to
\[
v_j^{n_j-1}\exp\left\{-v_j \sum_{k=1}^K m^{\dagger}_{j,k} J_k\right\}
\E\left[
\exp\left\{-v_j\sum_{k=1}^K m^{\star}_{j,k} J^{\star}_k\right\}
\right].
\]
Both full conditional densities are sampled  using a Metropolis-Hastings algorithm with random walk and an adaptive proposal distribution.

\subsubsection*{Updating $J^{\dagger}$ and $m^{\dagger}$}

The density of the full conditional distribution of $J^{\dagger}_k$ is proportional to
\[
\I\left(J^{\dagger}_k> \max\{u_{j,i}\vert c_{j,i}=k\}\right)\nu^{\star}\left(J^{\dagger}_k\right)
\exp\left\{-\sum_{l=1}^d v_l\sum_{r=1}^K m^{\dagger}_{l,r}J^{\dagger}_r\right\}
\]
where $n_{j,k}=\sum_{i=1}^{n_j} \I(c_{j,i}=k)$ and the full conditional density of $m^{\dagger}_{j,k}$
 is $\mbox{Ga}\left(\phi+n_{j,k}, 1 + v_j J^{\dagger}_k\right)$.

The elements of $A^{\dagger}$  are also updated using a reversible jump Metropolis-Hastings method with a birth and a death move which are proposed with equal probability. The birth move involves proposing a new jump $J^{\dagger}_{K+1}$ from a density proportional to $\nu^{\star}\left(J^{\dagger}_{K+1}\right)$ for $J^{\dagger}_{K+1}>L$
 and $m^{\dagger}_{1,K+1},\dots,m^{\dagger}_{d,K+1}\stackrel{i.i.d.}{\sim} f$. The death move proposes to delete an element of the set of jumps to which no observationsa are allocated
$
B=\left\{\left(J^{\dagger}_k,m^{\dagger}_{1,k},\dots,m^{\dagger}_{d,k}\right)\left\vert\sum_{j=1}^d n_{j,k}=0 \right. \right\}
$
 uniformly at random. If $b$ is the number of elements in $B$, the acceptance probability for the birth move is
\[
\min\left\{
1,\exp\left\{-\sum_{j=1}^d v_j J^{\dagger}_{K+1}m^{\dagger}_{j,K+1}\right\}
\frac{\int_L^{\infty} \nu^{\star}(z)\,dz}{b+1} \right\}
\]
and the acceptance probability if the $k$-th jump is proposed to be delete is
\[
\min\left\{
1,\exp\left\{\sum_{j=1}^d v_j J^{\dagger}_{k}m^{\dagger}_{j,k}\right\}
\frac{b} {\int_L^{\infty} \nu^{\star}(z)\,dz}\right\}.
\]

\subsubsection*{Updating $u$}

The full conditional distribution of $u_{j,i}$ is a uniform distribution on $\left(0,J^{\dagger}_{c_{j,i}}\right)$ for 
$i=1,\dots,n_j$ and
$j=1,\dots,d$. Let $\kappa$ be the $\min\{u_{j,i}\}$ from the previous iteration and  $\kappa^{\star}$ be the $\min\{u_{j,i}\}$ from the current iteration. If $\kappa^{\star}>\kappa$ then the jumps for which $J^{\dagger}_j<\kappa^{\star}$ are deleted. Otherwise, if $\kappa^{\star}<\kappa$, 
a Poisson distributed number of jumps with mean 
\[
\int_{\kappa^{\star}}^{\kappa} \nu^{\star}(z)\prod_{j=1}^d \int \exp\left\{- v_j m_j\right\}\, f(m_j)\,dm_j\,dz
\]
 are simulated from the density of $z$ proportional to
\[
\nu^{\star}(z)\prod_{j=1}^d \int \exp\left\{- v_j m_j z\right\}\, f(m_j)\,dm_j
,\qquad \kappa^{\star}<z<\kappa
\]
and $p\left(m^{\dagger}_j\right)\propto \exp\{-v_j z\}f\left(m^{\dagger}_j\right)$
Details on simulation for NCoRMs with Dirichlet process and normalized generalized gamma process marginals are provided in Appendix B.

\subsubsection*{Updating $\theta$}

The full conditional distribution of $\theta_k$ is 
\[
\alpha(\theta_k)
\prod_{j=1}^d 
\prod_{\{i\vert c_{j,i}=k\}}
k\left(y_{j,i}\vert \theta_k\right),\qquad k=1,\dots,K
\]

\subsubsection*{Updating the parameters of the NCoRM prior}

The full conditional distribution of the parameters of the NCoRM prior are  proportional to
\begin{align*}
&\prod_{j=1}^d \prod_{k=1}^K f\left(m_{j,k}^{\dagger}\right)\prod_{k=1}^K \nu^{\star}\left(J^{\dagger}_k\right)
\exp\left\{- \int_L^{\infty} \nu^{\star}(z) dz\right\}\\
&\times\exp\left\{
-\int_0^L 
 \left(1-\prod_{j=1}^{d} M_z^f(-v_j)\right)\nu^{\star}(z)\,dz
\right\}
\end{align*}

\subsubsection*{Updating $c_{j,i}$}

The full conditional distribution of $c_{j,i}$ is a discrete distribution with a finite number of possible states proportional to
\[
m^{\dagger}_{j, c_{j,i}}\I\left(J^{\dagger}_{c_{j,i}}>u_{j,i}\right)
k\left(y_{j,i}\vert \theta_{c_{j,i}}\right),\qquad 1,\dots,n_j,\quad j=1,\dots,d.
\]

\section{Illustrations}

The clinical studies
CALGB 8881
\citep{CALGB8881}
 and CALGB 9160 \citep{CALGB9160}  looked at the response of patients to different anticancer drug therapies. The response was white blood cell count (WBC) and patients had between four and 25 measurements taken over the course of the trial. The data was previously analysed by \cite{mulros04} who fit a nonlinear random effects model for the patient's response over time. The model assumes that the mean response at time $t$ with parameters $\theta=(z_1,z_2,z_3,\tau_1,\tau_2,\beta_0,\beta_1)$ is given by
\[
f(\theta,t)=\left\{
\begin{array}{ll}
z_1 & t<\tau_1\\
rz_1+ (1 - r) g(\theta,\tau_2) & \tau_1\leq t<\tau_2\\
g(\theta,t) & t\geq \tau_2
\end{array}
\right.
\]
where $r=(\tau_2-t)/(\tau_2-\tau_1)$ and $g(\theta, t)=z_2 + z_3/[1+\exp\{\beta_0 - \beta_1(t-\tau_2)\}]$. There were nine different combinations of the anticancer agent CTX, the drug GM-CSF and amifostine (AMOF) which are summarized in Table~\ref{t:data}.

\begin{table}[h!]
\begin{center}
\begin{tabular}{rrrrrr}\hline
Group & CTX & GM-CSF & AMOF & Study & Number of patients\\\hline
1 &   1.5  & 10.0  & 0 & 1 & 6\\
2  &  3.0  & 5.0  & 0 & 2 & 28\\
3   & 3.0  & 5.0  &  1 & 2 & 18\\
 4  & 3.0  & 2.5  & 0 & 1 & 6\\
5 &  3.0 &  5.0  & 0 & 1 & 6\\
6 &  3.0 &  10.0  & 0 & 1 & 6\\
7 &    4.5 &  5.0  & 0 & 1 & 12\\
8 &   4.5 &   10.0  & 0 & 1 & 10\\
9 &   6.0 & 5.0 &  0 & 1  & 6\\\hline
\end{tabular}
\end{center}
\caption{The levels of CTX (g $\mbox{m}^{-2}$), GM-CSF ($\mu$g k$\mbox{g}^{-1}$) and AMOF across the nine groups. CALGB 8881 is indicated as Study 1 and CALGB 9160 as Study 2.}\label{t:data}
\end{table}

Summaries of the data  are available as part of the \verb+DPpackage+ in  R where a  non-linear regression model is fitted with $f(\theta_{j,i},t)$ as the mean for the $i$-th patient in the $j$-th group. We will consider the differences in the distribution of the estimated values $\hat\theta_{j,i}$'s across the nine studies. It is assumed that
\[
\hat\theta_{j,i}\sim\N(\mu_{j,i},\Sigma_{j,i}),\qquad (\mu_{j,j},\Sigma_{j,j})\sim \tilde{p}_j
\]
where  $\tilde{p}_1,\dots,\tilde{p}_9$ are given a NCoRM process prior with independent 
$\Ga(\phi,1)$-distributed scores and Dirichlet process marginals. The centring measure $\alpha$ is $\mbox{N}(\mu\vert \bar{\hat\theta},100 \Sigma)\mbox{IW}(\Sigma\vert 14,4/9\times\hat\Sigma)$ where $\bar{\hat\theta}$ and $\hat\Sigma$ are the sample mean and the sample covariance matrix of $\hat\theta$. This implies a prior mean of $1/9\times \hat\Sigma$. 
 The parameter $\phi$ is given an exponential prior with mean 1.

\begin{figure}[h!]
\begin{center}
\includegraphics[trim=20mm 0mm 0mm 210mm, clip]{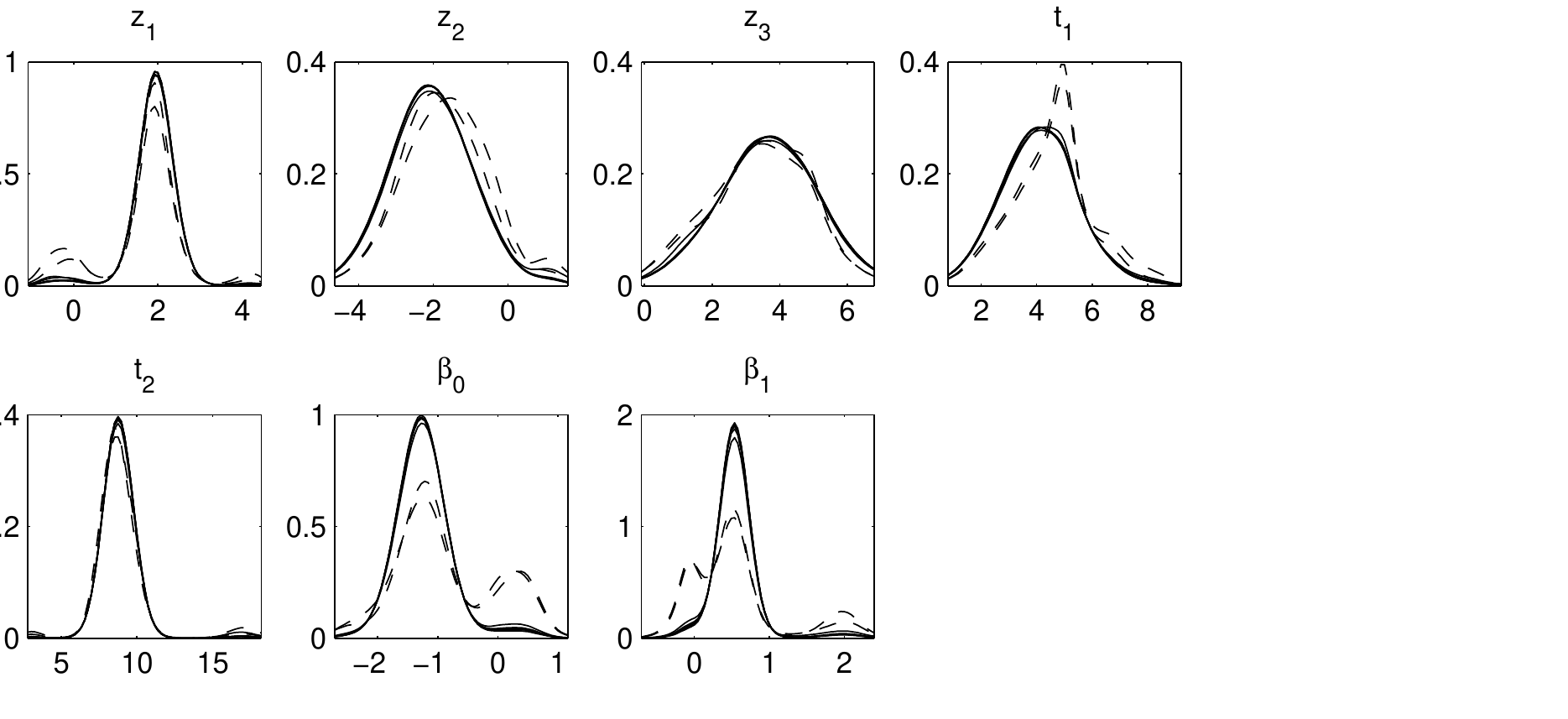}
\end{center}
\caption{The posterior mean marginal densities of each parameters in the CALGB example. The lines indicated a group in CALGB 8881 (solid line) and CALGB 9160 (dashed line).}\label{f:results}
\end{figure}
 The results of the analysis are illustrated  in Figure~\ref{f:results} which shows the posterior mean marginal density of each parameter. The results within each study are very similar with the main difference occurring between the two studies.
All densities are very similar for the parameters $z_1$, $z_2$, $z_3$ and $t_2$. There is a slight difference in the distribution for $t_1$ but much bigger differences for parameters $\beta_0$ and $\beta_1$. The results for CALGB 8881 are unimodal whereas CALGB9160 includes additional modes at 0.5 for $\beta_0$ and $-0.5$ and 2 for $\beta_1$. Figure~\ref{f:results2} shows the posterior mean joint density of $\beta_0$ and $\beta_1$ which shows a bimodal distribution for CALGB9160 with one mode at roughly $(-1.5, 0.5)$ (which is the mode for CALGB8881) and a second mode at roughly $(-0.5, 0)$. This suggests that CALGB9160 may contains two groups who responded differently.
The posterior median of $\phi$ was 1.03 with a 95\% highest posterior density region of $(0.46, 2.36)$.

\begin{figure}[h!]
\begin{center}
\includegraphics[trim=0mm 0mm 120mm 250mm, clip]{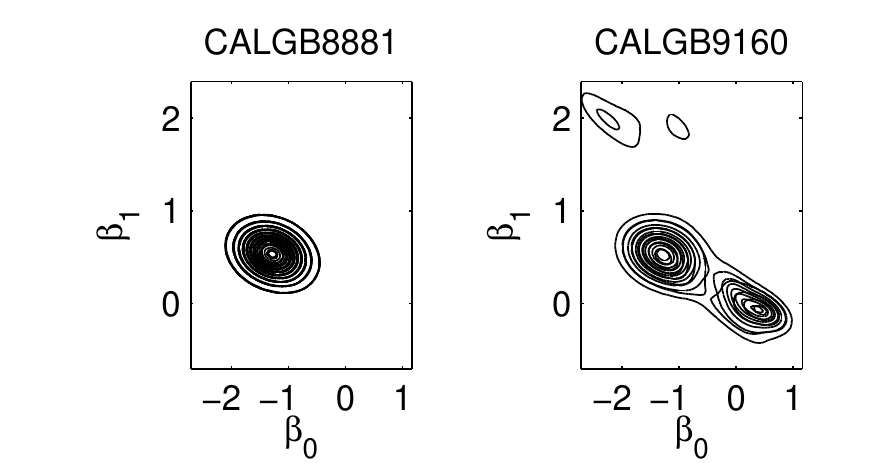}
\end{center}
\caption{The posterior mean joint densities of $\beta_0$ and $\beta_1$ in the CALGB example for the groups in CALGB 8881 and CALGB 9160.}\label{f:results2}
\end{figure}

\section{Discussion}

The modelling of dependent random measures has been an extremely active area of research for the past fifteen years beginning with the seminal work of MacEachern \citep{MacEachern}. Much of the work has concentrated on dependent random probability measures with several general approaches developed in the literature. Using the notation of (\ref{DDP}), initial work considered approaches where $w_i(x)=w_i$ and dependence is modelled through the atom location $\theta_i(x)$. This implies that cluster sizes will be similar for all values of $x$ and so leads to a specific form of dependence. Alternatively, 
many authors used $\theta_i(x)=\theta$ for all $x$ with dependence modelled through the weights; often using a 
 stick-breaking construction where $w_i(x) = V_i(x)\prod_{j<i}(1-V_j(x))$, see {\it e.g.} \citep{dun10} for a review. This usually leads to computationally tractable methods which either extend random truncation methods such as retrospective sampling 
\citep{Pap08}  or slice sampling \citep{Kalli11}, or develop truncation ideas for Dirichlet process mixtures 
\citep{IshJames}. However, stick-breaking approaches have some limitations for modelling. The construction implies a stochastic ordering so that $w_1(x)$ will tend to be the largest weight for all $x$. This can be inappropriate for some regression problems where we would like different component to have large weights for different values of $x$. The correlation is usually built on $V_j(x)$ and so $w_i(x)$ is a non-linear function of many correlated processes. This can lead to a dependence structure on $w_i(x)$ which is hard to interpret. Analytical results such as generalizations of the exchangeable partition probability function are usually impossible to derive for these priors. These methods can often be applied to problems where $\mathcal{X}$ is continuous or discrete. Other priors are restricted to a discrete $\mathcal{X}$.  One approach builds a hierarchy of nonparametric processes (see \cite{tehjord} for a review) leading from the seminal work of  \cite{tehjor06} on 
 hierarchical Dirichlet process (HDP).  For example, a two level hierarchical model could be constructed by assuming that the distributions for each group are conditionally independent draws from a nonparametric prior which is centred on a process which is itself given a nonparametric prior.  This leads to the same correlation {\it a priori} between the distribution for each value in $\mathcal{X}$ (although, more complicated hierarchical structures could be introduced to allow different correlation within subsets of $\mathcal{X}$).  Posterior simulation is usually implemented using the Chinese restaurant franchise algorithm. 
 
 The CoRM in its most general form is very flexible and allows both hierarchical and regression models. 
Normalized compound random measures includes many previously described priors which makes the links between these priors clearer. 
 This paper has concentrated on priors where the dimensions of the scores  are independent. 
 The tractability of these measures allows their properties to be derived and we concentrate on the class where the dimension of the scores are gamma distributed. 
 If the moment generating function of the marginal score distributions is available analytically, posterior computation for NCoRM mixture model can be carried out using an augmented P\'olya urn scheme or a slice sampler and several useful analytical expressions can be derived.  
 This restricts modelling to hierarchical type structures. More general, CoRM-type models where the scores are given by a regression are discussed by \cite{RanBlei15} who use a truncation of the infinite dimensional parameter and variational Bayes to make inference. In future work, we intend to extend both the P\'olya urn scheme and slice sampler to  regression models.

The compound random measure is defined using a completely random measure and a finite dimensional score distribution. For a given marginal process, the dependence between the distributions is controlled by the choice of finite dimensional score distribution. In this paper, we have concentrated on the case where the scores are independent and gamma distributed. 
This allows the dependence between the measures in different dimensions to be modelled by the shape parameter of the gamma distribution. In this case, we  show how compound random measures can be constructed with gamma, $\sigma$-stable and generalized gamma process marginals.  Importantly, the modelling of dependence between random measures can be achieved by the modelling of dependence between random variables and so greatly reduces the difficulty of  specifying a 
prior for a particular problem. Future work will consider studying these classes of compound random measures.

\section*{Acknowledgements}
Fabrizio Leisen was supported by the European Community's Seventh  Framework Programme [FP7/2007-2013] under grant agreement no: 630677 and Jim E. Griffin was supported by
 EPSRC Novel Technologies for Cross-disciplinary Research grant EP/I036575/I. The authors would like to acknowledge CALGB for the data used in the illustration.

\appendix

\section{Levy Copulas}

For the sake of illustration, we consider the 2-dimensional case.  
\begin{dfn}
A L\'evy copula is a function $C: [0,\infty]^2\rightarrow [0,\infty]$ such that
\begin{enumerate}
\item $C(y_1,0)=C(0,y_2)=0$ for any positive $y_1$ and $y_2$,
\item C has uniform margins, {\it i.e.} $C(y_1,\infty)=y_1$ and $C(\infty,y_2)=y_2$,
\item for all $y_1<z_1$ and $y_2<z_2$, $C(y_1,y_2)+C(z_1,z_2)-C(y_1,z_2)-C(y_2,z_1)\geq 0$.
\end{enumerate}
\end{dfn}
The definition in higher dimension is analogous (see \cite{ContTank}).
Let $U_i(x):=\int_x^\infty \nu_i(s)\,\ddr s$ be the $i$--th marginal tail integral associated with $\nu_i$. 
If both the copula $C$ and the marginal tail integrals are sufficiently smooth, then
  \[
  \rho_2(s_1,s_2)=\frac{\partial^2\, C(y_1, y_2)}{\partial y_1\partial
  y_2}\,\bigg|_{y_1=U_1(s_1), y_2=U_2(s_2)}\:
  \nu_1(s_1)\nu_2(s_2).
  \]

A wide range of dependence structures can be induced through L\'evy copulas. For example the independence case, {\it i.e.} $\int_{A\times B}\rho_2(s_1,s_2)\,\ddr s_1\,\ddr s_2=\int_A\nu_1(s_1)\,\ddr s_1+\int_B\nu_2(s_2)\,\ddr s_2$ for any $A$ and $B$ in $\Bcr(\R^+)$, corresponds to the L\'evy copula
\[
C_{\perp}(y_1,y_2)=y_1 \mathbb{I}_{\{\infty\}}(y_2)+y_2 \mathbb{I}_{\{\infty\}}(y_1).
\]
where $\mathbb{I}_A$ is the indicator function of the set $A$. 
On the other hand, the case of completely dependent CRMs  corresponds to
\[
C_{\parallel}(y_1,y_2)=\min \{y_1,y_2\}
\]
which yields a vector $(\tilde\mu_1,\tilde\mu_2)$ such that for any $x$ and $y$ in $\X$ either $\tilde\mu_i(\{x\})<\tilde\mu_i(\{y\})$ or $\tilde\mu_i(\{x\})>\tilde\mu_i(\{y\})$, for $i=1,2$, almost surely. Intermediate cases, between these two extremes, can be detected, for example, by relying on the {\it L\'evy-Clayton} copula defined by
\begin{equation}\label{LCcopula}
C_{\gamma}(y_1,y_2)=(y_1^{-\gamma}+y_2^{-\gamma})^{-\frac{1}{\gamma}}\qquad \gamma>0.
\end{equation}
with the parameter $\theta$ regulating the degree of dependence. It can be seen that
$\lim_{\gamma\rightarrow 0} C_{\gamma}= C_{\perp}$ and $\lim_{\gamma\rightarrow \infty} C_{\gamma}= C_{\parallel}.$

\section{Additional Results}

In the following theorem, the derivatives (up to a constant) of the Laplace exponent of a Compound random measure are provided. 
\begin{thm} Let
\begin{equation}
\label{grho}
g_{\rho}(q_1,\dots,q_d;\bm{\lambda})=\int_{(0,\infty)^d}s_1^{q_1}\cdots s_d^{q_d}e^{-\psi_{\rho,d}(\bm{\lambda})}\rho_d(s_1,\dots,s_d)d\bm{s}.
\end{equation}
Then,
$$g_{\rho}(q_1,\dots,q_d;\bm{\lambda})=(-1)^{q_1+\cdots+q_d}\int \nu^{\star}(z) \left(\prod_{j=1}^d \frac{\partial^{q_i}}{\partial \lambda_i^{q_i}}  M_z^f(-\lambda_i)\right)dz.$$
\end{thm}

\begin{proof}
\begin{equation*}
\begin{split}
g_{\rho}(q_1,\dots,q_d;\bm{\lambda})&=\int_{(0,\infty)^d}s_1^{q_1}\cdots s_d^{q_d}e^{-\lambda_1 s_1-\cdots-\lambda_d s_d}\rho_d(s_1,\dots,s_d)ds_1\cdots ds_d\\
&=\int \nu^{\star}(z) \left(\prod_{j=1}^d \int s_j^{q_1}e^{-\lambda_j s_j}z^{-1}f(s_j/z)ds_j\right)dz\\
&=\int \nu^{\star}(z) \left(\prod_{j=1}^d (-1)^{q_j} \frac{\partial^{q_j}}{\partial \lambda_j^{q_j}} \left(\int  e^{-\lambda_j s_j} z^{-1}f(s_j/z)ds_j\right)\right)dz\\
&=(-1)^{q_1+\cdots+q_d}\int \nu^{\star}(z) \left(\prod_{j=1}^d \frac{\partial^{q_j}}{\partial \lambda_j^{q_j}}  M_z^f(-\lambda_j)\right)dz\\
\end{split}
\end{equation*}
\end{proof}

\section{Proofs}\label{app}

\subsubsection*{Proof of Theorem 3.1}

%

\begin{equation*}
\begin{split}
\psi_{\rho,d}(\lambda_1,\dots,\lambda_d)&=\int_{[0,+\infty]^d} (1-e^{-\lambda_1 s_1-\cdots-\lambda_d s_d})\rho_d(s_1,\dots,s_d)ds_1\dots ds_d\\
&=\int_{[0,+\infty]^d} (1-e^{-\lambda_1 s_1-\cdots-\lambda_d s_d})\int z^{-d}\prod_{j=1}^d f(s_j/z)\nu^{\star}(z)\,dzds_1\dots ds_d\\
&=\int z^{-d}\left(\int_{[0,+\infty]^d}(1-e^{-\lambda_1 s_1-\cdots-\lambda_d s_d})\prod_{j=1}^d f(s_j/z)ds_1\dots ds_d\right)\nu^{\star}(z)\,dz\\
&=\int z^{-d}\left(\int_{[0,+\infty]^d}\left(\prod_{j=1}^d f(s_j/z)-e^{-\lambda_1 s_1-\cdots-\lambda_d s_d}\prod_{j=1}^d f(s_j/z)\right)ds_1\dots ds_d\right)\nu^{\star}(z)\,dz\\
&=\int z^{-d}\left(z^d-\int_{[0,+\infty]^d}\prod_{j=1}^d e^{-\lambda_j s_j}f(s_j/z)ds_1\dots ds_d\right)\nu^{\star}(z)\,dz\\
&=\int z^{-d}\left(z^d-\prod_{j=1}^d\int_{[0,+\infty]} e^{-\lambda_j s_j}f(s_j/z)ds_j\right)\nu^{\star}(z)\,dz\\
&=\int z^{-d}\left(z^d-\prod_{j=1}^d z M_z^f(-\lambda_j)\right)\nu^{\star}(z)\,dz\\
&=\int \left(1-\prod_{j=1}^{d} M_z^f(-\lambda_j)\right)\nu^{\star}(z)\,dz
\end{split}
\end{equation*}

\subsubsection*{Proof of Theorem 4.1} \textbf{Gamma marginals}.
 From 3.383.4 of \cite{Grad} follows the thesis. Indeed,  
\begin{equation*}
\begin{split}
\rho_d(s_1,\dots,s_d)&=\int \nu^{\star}(z) z^{-d}\prod_{j=1}^d  f(s_j/z)dz\\
&=\int_0^1 z^{-1}(1-z)^{\phi-1}z^{-d}\frac{\left(\prod_{j=1}^d s_j\right)^{\phi-1}}{[\Gamma(\phi)]^d}z^{-d\phi+d}e^{-\frac{\bm{|s|}}{z}}dz\\
&=\frac{\left(\prod_{j=1}^d s_j\right)^{\phi-1}}{[\Gamma(\phi)]^d}\int_0^1 z^{-d\phi-1}(1-z)^{\phi-1}e^{-\frac{|\bm{s}|}{z}}dz\\
&=\frac{\left(\prod_{j=1}^d s_j\right)^{\phi-1}}{[\Gamma(\phi)]^d}\int_1^{+\infty} t^{(d-1)\phi+1}\left(t-1\right)^{\phi-1}e^{-|s|t}dt\\
&=\frac{\left(\prod_{j=1}^d s_j\right)^{\phi-1}}{[\Gamma(\phi)]^{d-1}}|\bm{s}|^{-\frac{d\phi+1}{2}}e^{-\frac{|\bm{s}|}{2}}W_{\frac{(d-2)\phi+1}{2},-\frac{d\phi}{2}}(|\bm{s}|)
\end{split}
\end{equation*}
\textbf{$\sigma$-stable marginals.} In this case,

\begin{equation*}
\begin{split}
\rho_d(s_1,\dots,s_d)&=\int_0^{+\infty}z^{-\sigma-1}\frac{\Gamma(\phi)}{\Gamma(\sigma)\Gamma(1-\sigma)} z^{-d}\frac{\left(\prod_{j=1}^d s_j\right)^{\phi-1}}{[\Gamma(\phi)]^d}z^{-d\phi+d}e^{-\frac{\bm{|s|}}{z}}dz\\
&=\frac{1}{\Gamma(\sigma)\Gamma(1-\sigma)}\frac{\left(\prod_{j=1}^d s_i\right)^{\phi-1}}{[\Gamma(\phi)]^{d-1}}\int_0^{+\infty} z^{-d\phi-\sigma-1}e^{-\frac{\bm{|s|}}{z}}dz\\
&=\frac{\Gamma(\sigma+d\phi)}{\Gamma(\sigma)\Gamma(1-\sigma)}\frac{\left(\prod_{j=1}^d s_j\right)^{\phi-1}}{[\Gamma(\phi)]^{d-1}}\bm{|s|}^{-\sigma-d\phi}
\end{split}
\end{equation*}

\subsubsection*{Proof of Corollary 4.1} \textbf{Gamma marginals}.
First of all, note that the Whittaker function could be expressed in terms of a Kummer confluent hypergeometric function, 
$$W_{\frac{(d-2)\phi+1}{2},-\frac{d\phi}{2}}(|\bm{s}|)=e^{-\frac{|\bm{s}|}{2}}|\bm{s}|^{-\frac{d\phi-1}{2}}U(-(d-1)\phi,-d\phi+1,|\bm{s}|)$$
and thus
$$\rho_d(s_1,\dots,s_d)=\frac{\left(\prod_{i=1}^d s_i\right)^{\phi-1}}{[\Gamma(\phi)]^{d-1}}|\bm{s}|^{-d\phi}e^{-|\bm{s}|}U(-(d-1)\phi,-d\phi+1,|\bm{s}|)$$
In the special case of $\phi=1$, we get the following multivariate L\'evy intensity
$$\rho_d(s_1,\dots,s_d)=|\bm{s}|^{-d}e^{-|\bm{s}|}U(-d+1,-d+1,|\bm{s}|)$$
and from 13.2.7 of \cite{NIST}
\begin{equation*}
\begin{split}
\rho_d(s_1,\dots,s_d)&=|\bm{s}|^{-d}e^{-|\bm{s}|}(-1)^{d-1}\sum_{j=0}^{d-1}\binom{d-1}{j} (-1)^{j} (-d+1+j)_{d-1-j} |\bm{s}|^j \\
&=|\bm{s}|^{-d}e^{-|\bm{s}|}\sum_{j=0}^{d-1}\binom{d-1}{j} j!|\bm{s}|^{d-1-j} \\
&=\sum_{j=0}^{d-1}\frac{(d-1)!}{(d-1-j)!}|\bm{s}|^{-j-1}e^{-|\bm{s}|}\\
\end{split}
\end{equation*}
\textbf{$\sigma$-stable marginals}. The second part of the proof is straightforward and doesn't require additional algebra.

\subsubsection*{Proof of Theorem 4.2}
From Equation \eqref{MGM} it follows
$$\psi_{\rho,d}(\bm{\tilde{\lambda}},\bm{n})=\int \left(1-\prod_{i=1}^{l} \frac{1}{(1+z\tilde{\lambda}_i)^{n_i\phi}}\right)\nu^{\star}(z)dz$$
since $M_z^f(-\tilde{\lambda}_i)=\frac{1}{(1+z\tilde{\lambda}_i)^{\phi}}$ under the hypothesis of independent Gamma distributed scores. The conclusion follows by noting that $\frac{\partial^{(n_i-1)\phi}}{\partial^{(n_i-1)\phi}\tilde{\lambda}_i}(\tilde{\lambda}_i^{n_i\phi-1})=\frac{\Gamma(n_i\phi)}{\Gamma(\phi)}\tilde{\lambda}_i^{\phi-1}$ and 
$$\frac{\partial^{(n_i-1)\phi}}{\partial^{(n_i-1)\phi}\tilde{\lambda}_i}\left(\frac{\tilde{\lambda}_i^{n_i\phi-1}}{(1+z\tilde{\lambda}_i)^{\phi}}\right)=\frac{\Gamma(n_i\phi)}{\Gamma(\phi)}\frac{\tilde{\lambda}_i^{\phi-1}}{(1+z\tilde{\lambda}_i)^{n_i\phi}}$$

The last equality follows from a simple application of the Leibniz's formula, indeed

\begin{equation*}
\begin{split}
\frac{\partial^{(n_i-1)\phi}}{\partial^{(n_i-1)\phi}\tilde{\lambda}_i}\left(\frac{\tilde{\lambda}_i^{n_i\phi-1}}{(1+z\tilde{\lambda}_i)^{\phi}}\right)&=\sum_{j=0}^{(n_i-1)\phi}\binom{(n_i-1)\phi}{j} \frac{(j+\phi-1)!}{(\phi-1)!}\frac{ (-1)^jz^j }{(1+z\tilde{\lambda}_i)^{j+\phi}}\frac{(n_i\phi-1)!}{(j+\phi-1)!}\tilde{\lambda}_i^{j+\phi-1}\\
&=\frac{\Gamma(n_i\phi)}{\Gamma(\phi)}\frac{\tilde{\lambda}_i^{\phi-1}}{(1+z\tilde{\lambda}_i)^{\phi}}\sum_{j=0}^{(n_i-1)\phi}\binom{(n_i-1)\phi}{j}\left(\frac{-z\tilde{\lambda}_i }{1+z\tilde{\lambda}_i}\right)^j\\
&=\frac{\Gamma(n_i\phi)}{\Gamma(\phi)}\frac{\tilde{\lambda}_i^{\phi-1}}{(1+z\tilde{\lambda}_i)^{\phi}}\left(\frac{-z\tilde{\lambda}_i }{1+z\tilde{\lambda}_i}+1\right)^{(n_i-1)\phi}\\
&=\frac{\Gamma(n_i\phi)}{\Gamma(\phi)}\frac{\tilde{\lambda}_i^{\phi-1}}{(1+z\tilde{\lambda}_i)^{n_i\phi}}\\
\end{split}
\end{equation*}
Thus, 
\begin{equation*}
\begin{split}
\psi_{\rho,d}(\bm{\tilde{\lambda}},\bm{n})&=\int \left(1-\prod_{i=1}^{l} \frac{1}{(1+z\tilde{\lambda}_i)^{n_i\phi}}\right)\nu^{\star}(z)dz\\
&=\frac{1}{\prod_{i=1}^l \tilde{\lambda}_i^{\phi-1}}\int \left(\prod_{i=1}^l \tilde{\lambda}_i^{\phi-1}-\prod_{i=1}^{l} \frac{\tilde{\lambda}_i^{\phi-1}}{(1+z\tilde{\lambda}_i)^{n_i\phi}}\right)\nu^{\star}(z)dz\\
&=\frac{[\Gamma(\phi)]^l}{\prod_{i=1}^l [\tilde{\lambda}_i^{\phi-1}\Gamma(n_i\phi)]}\left(\prod_{i=1}^l\frac{\partial^{(n_i-1)\phi}}{\partial^{(n_i-1)\phi}\tilde{\lambda}_i}\right)\int \left(\prod_{i=1}^l \tilde{\lambda}_i^{n_i\phi-1}-\prod_{i=1}^{l} \frac{\tilde{\lambda}_i^{n_i\phi-1}}{(1+z\tilde{\lambda}_i)^{\phi}}\right)\nu^{\star}(z)dz\\
&=\frac{[\Gamma(\phi)]^l}{\prod_{i=1}^l [\tilde{\lambda}_i^{\phi-1}\Gamma(n_i\phi)]}\left(\prod_{i=1}^l\frac{\partial^{(n_i-1)\phi}}{\partial^{(n_i-1)\phi}\tilde{\lambda}_i}\right)\left(\Upsilon^{\phi}_l(\bm{\tilde{\lambda}})\prod_{i=1}^l \tilde{\lambda}_i^{n_i\phi-1}\right)\\
\end{split}
\end{equation*}

\subsubsection*{Proof of Theorem 4.3}
Let
$$B_{\phi,l}=\{\bm{k}^*\in\{0,1,\ldots,\phi\}^l:\: |\bm{k}^*|=\phi\}\qquad \phi\geq j.$$
First of all, note that
$$\sum_{i=1}^l a_i(\bm{\tilde{\lambda}})=1$$
and 
$$\prod_{i=1}^{l} \frac{1}{(1+z\tilde{\lambda}_i)}=\sum_{i=1}^l \frac{a_i(\bm{\tilde{\lambda}})}{(1+z\tilde{\lambda}_i)}$$
Thus,
\begin{equation*}
\begin{split}
\Upsilon_l^{\phi}(\bm{\tilde{\lambda}})&=\int \left[\left(\sum_{i=1}^l a_i(\bm{\tilde{\lambda}})\right)^{\phi}-\left(\sum_{i=1}^l \frac{a_i(\bm{\tilde{\lambda}})}{(1+z\tilde{\lambda}_i)}\right)^{\phi}\right]\nu^{\star}(z)dz\\
&=\sum_{\bm{k}^*\in B_{\phi,l}} \binom{\phi}{k_1^*,\cdots,k_l^*}a_1^{k_1^*}(\bm{\tilde{\lambda}})\cdots a_l^{k_l^*}(\bm{\tilde{\lambda}})I(k_1^*,\dots,k_l^*;\bm{\tilde{\lambda}})\\
\end{split}
\end{equation*}
where 
$$I(k_1^*,\dots,k_l^*;\bm{\tilde{\lambda}})=\int \left(1-\prod_{i=1}^{l} \frac{1}{(1+z\tilde{\lambda}_i)^{k_i^*}}\right)\nu^{\star}(z)dz$$
Since some of the $k^*$'s could be zero then some terms could disappear in the expression above. For this reason, it's more convenient to write $\Upsilon_l^{\phi}(\bm{\tilde{\lambda}})$ as a sum over the set $A_{\phi,j}$ instead of a sum over $B_{\phi,l}$. Thus,
\begin{equation*}
\begin{split}
\Upsilon_l^{\phi}(\bm{\tilde{\lambda}})&=\sum_{j=1}^l \sum_{\bm{k}\in A_{\phi,j}} \binom{\phi}{k_1,\cdots,k_j}\sum_{0<i_1<i_2<\cdots<i_j\leq l} a_{i_1}^{k_1}(\bm{\tilde{\lambda}})\cdots a_{i_j}^{k_j}(\bm{\tilde{\lambda}})C(i_1,\dots,i_j;\bm{k};\bm{\tilde{\lambda}})\\
\end{split}
\end{equation*}
where 
$$C(i_1,\dots,i_j;\bm{k};\bm{\tilde{\lambda}})=\int \left(1-\prod_{h=1}^{j} \frac{1}{(1+z\tilde{\lambda}_{i_h})^{k_h}}\right)\nu^{\star}(z)dz$$
If $j>\phi$ then $A_{\phi,j}$ is the empty set. Thus we can resort the above sum as 
\begin{equation*}
\begin{split}
\Upsilon_l^{\phi}(\bm{\tilde{\lambda}})&=\sum_{j=1}^{\phi} \sum_{\bm{k}\in A_{\phi,j}} \binom{\phi}{k_1,\cdots,k_j}\sum_{0<i_1<i_2<\cdots<i_j\leq l} a_{i_1}^{k_1}(\bm{\tilde{\lambda}})\cdots a_{i_j}^{k_j}(\bm{\tilde{\lambda}})C(i_1,\dots,i_j;\bm{k};\bm{\tilde{\lambda}})\\
\end{split}
\end{equation*}
Let $|\bm{y}|_{\bm{i}}^j=y_{i_1}+\cdots+y_{i_j}$. Note that
\begin{equation*}
\begin{split}
\int \nu^{\star}(z)dz&=\int \int_{[0,+\infty]^j} \prod_{h=1}^j \frac{z^{-k_h}}{\Gamma(k_h)}y_{i_h}^{k_h-1}e^{-\frac{y_{i_h}}{z}} d\bm{y}\nu^{\star}(z)dz\\
&=\int_{[0,+\infty]^j}\frac{\Gamma(\phi)}{\left(|\bm{y}|_{\bm{i}}^j\right)^{\phi-1}}\left(\prod_{h=1}^j \frac{y_{i_h}^{k_h-1}}{\Gamma(k_h)}\right) \int \frac{\left(|\bm{y}|_{\bm{i}}^j\right)^{\phi-1}}{z^{\phi}\Gamma(\phi)}e^{-\frac{|\bm{y}|_{\bm{i}}^j}{z}} \nu^{\star}(z)dzd\bm{y}\\
&=\int_{[0,+\infty]^j}\frac{\Gamma(\phi)}{\left(|\bm{y}|_{\bm{i}}^j\right)^{\phi-1}}\left(\prod_{h=1}^j \frac{y_{i_h}^{k_h-1}}{\Gamma(k_h)}\right)\nu(|\bm{y}|_{\bm{i}}^j)d\bm{y}\\\\
\end{split}
\end{equation*}
In a similar fashion,
\begin{equation*}
\begin{split}
\int \left(\prod_{h=1}^{j} \frac{1}{(1+z\tilde{\lambda}_{i_h})^{k_h}}\right)\nu^{\star}(z)dz&=\int_{[0,+\infty]^j}\frac{\Gamma(\phi)e^{-\sum_{h=1}^j \tilde{\lambda}_{i_h}y_{i_h}}}{\left(|\bm{y}|_{\bm{i}}^j\right)^{\phi-1}}\left(\prod_{h=1}^j \frac{y_{i_h}^{k_h-1}}{\Gamma(k_h)}\right)\nu(|\bm{y}|_{\bm{i}}^j)d\bm{y}\\\\
\end{split}
\end{equation*}
and thus,
$$C(i_1,\dots,i_j;\bm{k};\bm{\tilde{\lambda}})=\int_{[0,+\infty]^j} \left(1-e^{-\sum_{h=1}^j \tilde{\lambda}_{i_h}y_{i_h}}\right)\frac{\Gamma(\phi)}{\left(|\bm{y}|_{\bm{i}}^j\right)^{\phi-1}}\left(\prod_{h=1}^j \frac{y_{i_h}^{k_h-1}}{\Gamma(k_h)}\right)\nu(|\bm{y}|_{\bm{i}}^j)d\bm{y}$$
The change of variables $\rho=|\bm{y}|_{\bm{i}}^j$ and $y_{i_h}=\rho z_h$, $h=1,\dots,j-1$ leads to 
\begin{equation*}
\begin{split}
C(i_1,\dots,i_j;\bm{k};\bm{\tilde{\lambda}})&=\int_{\Delta_{j-1}} \left(\Gamma(\phi)\prod_{h=1}^j \frac{z_h^{k_h-1}}{\Gamma(k_h)}\right)\int_0^{+\infty}\left(1-e^{-\rho\Lambda(\bm{\tilde{\lambda}},\bm{z})}\right)\nu(\rho)d\rho d\bm{z}\\
&=\int_{\Delta_{j-1}} \Gamma(\phi)\left((1-\sum_{h=1}^{j-1}z_h)^{k_j}\prod_{h=1}^{j-1} \frac{z_h^{k_h-1}}{\Gamma(k_h)}\right)\psi\left(\Lambda(\bm{\tilde{\lambda}},\bm{z})\right)d\bm{z}
\end{split}
\end{equation*}

\subsubsection*{Proof of Theorem 4.4}

Let $U(x)=\int _x^{+\infty} \nu(x)dx$ be the tail integral of the marginal L\'evy intensity and let 
$$U(x_1,\dots,x_d)=\int_{x_1}\dots\int_{x_d}\nu(s_1,\dots,s_d)ds_1\cdots ds_d.$$
From Theorem 5.3 in Cont and Tankov, exists only one copula C such that 
$$U(x_1,\dots,x_d)=C(U(x_1),\dots,U(x_d)).$$
It's easy to see that 
$$U(x_1,\dots,x_d)=\int \nu^{\star}(z)\prod_{j=1}^d (1-F(z^{-1}x_j))dz$$
and this proves the thesis.

\subsubsection*{Proof of Theorem 4.5}
In a similar fashion of \cite{LLS}, it's easy to see that
$$
\E\left[\prod_{i=1}^d\{\tilde\mu_i(A)\}^{q_i}\right]=e^{-\alpha(A)\psi(0,...,0)}\,q_1!\cdots q_d!\: \sum_{k=1}^{\mid \bm{q}\mid}[\alpha(A)]^k\:\times\:$$
$$\hspace{3cm}\times  \sum_{j=1}^{\mid \bm{q}\mid} \:\sum_{p_j(\bm{q},k)}\:\prod_{i=1}^j \frac{1}{\eta_i!(s_{1,i}!\cdots s_{d,i}!)^{\eta_i}}\left(g_{\nu}(s_{1,i},\dots,s_{d,i},0,\dots,0)\right)^{\eta_i}
$$
In the gamma case, it's easy to see that $\psi(0,...,0)=0$ and  
$$g_{\nu}(s_{1,i},\dots,s_{d,i},0,\dots,0)=\int \nu^{\star}(z) \prod_{j=1}^d (\phi)_{s_{j,i}}z^{s_{j,i}}dz$$
and this concludes the proof.

\section{Additional Details of Computational Methods}\label{app2}

In the update of $u$ in the Gibbs sampler, it is necessary to sample from the density proportional to
 simulated from the density of $z$ proportional to
\[
\nu^{\star}(z)\prod_{j=1}^d \int \exp\left\{- v_j m_j z\right\}\, f(m_j)\,dm_j
,\qquad \kappa^{\star}<z<\kappa.
\]
If a NCoRM with a $\mbox{Ga}(\phi, 1)$ score distribution and Dirichlet process marginals is used, this density is proportional to
\[
z^{-1}(1-z)^{\phi-1}\prod_{j=1}^d (1 + v_i z)^{-\phi},\qquad \kappa^{\star}<z<\kappa.
\]
A rejection sampler is used with rejection envelope proportional to $z^{-1}(1-z)^{\phi-1},  \ \kappa^{\star}<z<\kappa$. The acceptance probability is
\[
\prod_{j=1}^d \left(\frac{1 + v_j \kappa^{\star}}{1 + v_j z}\right)^{\phi}.
\]
This rejection envelope is non-standard and can be sampled using a rejection sampler with the 
envelope
\[
g(z) =\left\{ \begin{array}{rll}
(1-z)^{\phi-1},  & \kappa^{\star}<z<\kappa, & \mbox{if }\phi<1\\
z^{-1}, &  \kappa^{\star}<z<\kappa, & \mbox{if }\phi>1
\end{array}\right..
\]
If a NCoRM with a $\mbox{Ga}(\phi, 1)$ score distribution and normalized generalized gamma process marginals with $a=1$ is used, this density is proportional to
\[
z^{-1-\sigma}(1-z)^{\sigma+\phi-1}\prod_{j=1}^d (1 + v_i z)^{-\phi},\qquad \kappa^{\star}<z<\kappa.
\]
A rejection sampler is used with rejection envelope proportional to $z^{-1-\sigma}(1-z)^{\sigma+\phi-1},  \ \kappa^{\star}<z<\kappa$. The acceptance probability is
\[
\prod_{j=1}^d \left(\frac{1 + v_j \kappa^{\star}}{1 + v_j z}\right)^{\phi}.
\]
This rejection envelope is non-standard and can be sampled using a rejection sampler with the 
envelope
\[
g(z) =\left\{ \begin{array}{rll}
(1-z)^{\sigma+\phi-1}, & \kappa^{\star}<z<\kappa, & \mbox{if }\sigma+\phi<1\\
z^{-1-\sigma}, & \kappa^{\star}<z<\kappa, & \mbox{if }\sigma+\phi>1
\end{array}\right..
\]

\end{document}